\documentclass[11pt]{article}

\usepackage[margin=1in]{geometry}
\usepackage{amsmath, amsthm, amssymb, thmtools, mleftright}
\usepackage{xspace, paralist, enumitem, multirow, booktabs}
\usepackage[usenames, dvipsnames]{xcolor}
\usepackage{times, mathptmx, bm}
\usepackage[T1]{fontenc}
\usepackage{graphicx}
\usepackage{algorithm}
\usepackage[noend]{algpseudocode}
\usepackage{url}
\usepackage[colorlinks, linkcolor=BrickRed, citecolor=blue]{hyperref}
\usepackage{cite, cleveref}

\allowdisplaybreaks[1]

\makeatletter
\def\th@plain{%
  \thm@notefont{}% same as heading font
  \itshape % body font
}
\def\th@definition{%
  \thm@notefont{}% same as heading font
  \normalfont % body font
}
\makeatother
\newtheorem{theorem}{Theorem}
\newtheorem{lemma}{Lemma}[section]

\newtheorem{proposition}[theorem]{Proposition}
\newtheorem*{reminder}{Reminder}
\theoremstyle{definition}  \newtheorem{definition}[lemma]{Definition}

\theoremstyle{remark}  \newtheorem*{remark}{Remark}
\newlist{thmparts}{enumerate}{1}
\setlist[thmparts]{labelindent=\parindent,leftmargin=*,itemsep=2pt,font=\normalfont,label=(\thetheorem.\arabic*)}
\Crefname{thmpartsi}{Subresult}{Subresults} % any better name? subresult? subtheorem?
\declaretheoremstyle[%
  spaceabove=-6pt,%
  spacebelow=6pt,%
  headfont=\normalfont\itshape,%
  postheadspace=1em,%
  qed=\qedsymbol%
]{mystyle} 
\declaretheorem[name={Proof},style=mystyle,unnumbered,
]{prf}
\RequirePackage{caption, float}
\algrenewcomment[1]{\hfill \textcolor{blue}{$\triangleright$ #1}}
\algnewcommand{\LineComment}[1]{\State \textcolor{blue}{$\triangleright$ #1}}
\algdef{SE}[SUBALG]{Indent}{EndIndent}{}{\algorithmicend\ }%
\algtext*{Indent}
\algtext*{EndIndent}

\newcommand{\EE}{\mathbb{E}}

\newcommand{\cA}{\mathcal{A}}

\newcommand{\cE}{\mathcal{E}}
\newcommand{\cF}{\mathcal{F}}

\newcommand{\cT}{\mathcal{T}}

\newcommand{\bx}{\mathbf{x}}
\newcommand{\by}{\mathbf{y}}

\newcommand{\tT}{\widetilde{T}}
\newcommand{\tO}{\widetilde{O}}

\DeclareMathOperator{\col}{col}

\DeclareMathOperator{\codeg}{codeg}

\DeclareMathOperator{\oracle}{orcl}

\DeclareMathOperator{\polylog}{polylog}

\DeclareMathOperator{\Var}{Var}
\DeclareMathOperator{\Cov}{Cov}
\DeclareMathOperator{\ID}{ID}
\DeclareMathOperator{\Nhd}{N}
\DeclareMathOperator{\Bin}{Bin}

\newcommand{\Tkheavy}{T_k^H}
\newcommand{\Tklight}{T_k^L}
\newcommand{\Perr}{P_{\text{err}}}
\newcommand{\scor}{s_{\text{cor}}}  % sum quantifying inter-simplex correlation
\newcommand{\reldeg}[2]{\deg(#2 \mid #1)}  % relative degree
\newcommand{\apex}{z}      % apex of a simplex
\newcommand{\nsimp}{\tau}  % number of simplices labeled by a particular base edge
\newcommand{\Nsimp}{M}     % number of simplices that include a particular edge
\newcommand{\nhyp}[2]{{#1}\downarrow{#2}}  % neighborhood hypergraph
\newcommand{\flavor}[2]{#2^{\ang{#1}}}  % flavored vertices 

\newcommand{\disj}{\textsc{disj}\xspace}
\newcommand{\udisj}{\textsc{udisj}\xspace}

\newcommand{\idx}{\textsc{index}\xspace}

\newcommand{\gapdisj}{\textsc{gap-disj}\xspace}
\newcommand{\SimplexEst}{\textsc{simplex-est}\xspace}
\newcommand{\SimplexDist}{\textsc{simplex-dist}\xspace}
\newcommand{\SimplexSep}{\textsc{simplex-sep}\xspace}
\newcommand{\heavy}{\textsc{heavy}\xspace}
\newcommand{\light}{\textsc{light}\xspace}

\newcommand{\ang}[1]{{\langle{#1}\rangle}}
\newcommand{\ceil}[1]{{\lceil{#1}\rceil}}
\newcommand{\floor}[1]{{\lfloor{#1}\rfloor}}
\renewcommand{\b}{\{0,1\}}

\newcommand{\eps}{\varepsilon}
\newcommand{\setm}{\smallsetminus}
\DeclareMathOperator{\nperp}{\not\perp}

\newcommand{\ontop}[2]{\genfrac{}{}{0pt}{2}{#1}{#2}}
\makeatletter
\newcommand{\tpmod}[1]{{\@displayfalse\pmod{#1}}}
\makeatother

\makeatletter
\newcommand\squeezepar{\@startsection{paragraph}{4}{\z@}{1.5ex \@plus1ex \@minus.2ex}{-1em}{\normalfont\normalsize\bfseries}}
\makeatother
\begin{document}

\title{Counting Simplices in Hypergraph Streams}

\date{}

\author{%
  Amit Chakrabarti
  \thanks{Department of Computer Science, Dartmouth College.
  Work supported in part by NSF under award 1907738.}
  \and
  Themistoklis Haris $^\fnsymbol{footnote}$
}

\maketitle

\thispagestyle{empty}

\begin{abstract}

\noindent
We consider the problem of space-efficiently estimating the number of
\emph{simplices} in a hypergraph stream. This is the most natural hypergraph
generalization of the highly-studied problem of estimating the number of
\emph{triangles} in a graph stream. Our input is a $k$-uniform hypergraph $H$
with $n$ vertices and $m$ hyperedges, each hyperedge being a $k$-sized subset
of vertices. A $k$-simplex in $H$ is a subhypergraph on $k+1$ vertices $X$
such that all $k+1$ possible hyperedges among $X$ exist in $H$. The goal is to
process the hyperedges of $H$, which arrive in an arbitrary order as a data
stream, and compute a good estimate of $T_k(H)$, the number of $k$-simplices
in $H$.

We design a suite of algorithms for this problem. As with triangle-counting in
graphs (which is the special case $k=2$), sublinear space is achievable but only under a
promise of the form $T_k(H) \ge T$. Under such a promise, our algorithms use
at most four passes and together imply a space bound of 
\[
  O\left( \varepsilon^{-2} \log\delta^{-1} \mathop{\text{polylog}} n \cdot 
    \min\left\{ \frac{m^{1+1/k}}{T}, \frac{m}{T^{2/(k+1)}} \right\} \right)
\]
for each fixed $k \ge 3$, in order to guarantee an estimate within
$(1\pm\varepsilon)T_k(H)$ with probability at least $1-\delta$. We also give a
simpler $1$-pass algorithm that achieves $O\left(\varepsilon^{-2}\log\delta^{-1}\log
n\cdot (m/T) \left( \Delta_E + \Delta_V^{1-1/k} \right)\right)$ space, where
$\Delta_E$ (respectively, $\Delta_V$) denotes the maximum number of
$k$-simplices that share a hyperedge (respectively, a vertex), which generalizes
a previous result for the $k=2$ case. We complement these algorithmic
results with space lower bounds of the form $\Omega(\varepsilon^{-2})$,
$\Omega(m^{1+1/k}/T)$, $\Omega(m/T^{1-1/k})$ and $\Omega(m\Delta_V^{1/k}/T)$
for multi-pass algorithms and $\Omega(m\Delta_E/T)$ for $1$-pass algorithms,
which show that some of the dependencies on parameters in our upper bounds are
nearly tight. Our techniques extend and generalize several different ideas
previously developed for triangle counting in graphs, using appropriate
innovations to handle the more complicated combinatorics of hypergraphs.

\end{abstract}

% Keywords: Data streaming, graph algorithms, hypergraphs, sub-linear algorithms, triangle counting

\newpage
\addtocounter{page}{-1}

% The actual sections of the paper, ideally one file per section
\section{Introduction} \label{sec:intro}

Estimating the number of triangles in a massive input graph is a fundamental
algorithmic problem that has attracted over two decades of intense
research~\cite{%
AlonYZ97,BarYossefKS02,JowhariG05,BuriolFLMS06,SuriV11,TsourakakisKM11,%
KolountzakisMPT12,PaghT12,BravermanOV13,PavanTTW13,JhaSP15,BulteauFKP16,%
McGregorVV16,BeraC17,CormodeJ17,KallaugherP17,JayaramK21}.
It is easy to see why. On the one hand, the problem arises in applications
where a complex real-world network is naturally modeled as a graph and the
number of triangles is a crucial statistic of the network. Such applications
are found in many different domains, such as social
networks~\cite{ChristakisF10,LeskovecHK10,Newman03}, the web
graph~\cite{BecchettiBCG08,WuCWWZW18}, and biological
networks~\cite{RougemontH03}; see~\cite{TsourakakisKM11} for a more detailed
discussion of such applications. On the other hand, a triangle is perhaps
{\em the} most basic nontrivial pattern in a graph and as such, triangle counting is
a problem with a rich theory and connections to many areas within 
computer science~\cite{AlonYZ97,EdenLRS15,GronlundP18,NgoPRR18,DurajKPW20} and
combinatorics~\cite{Mantel1907,Fisher89,KimV04,Razborov08}.

In this work, we study the natural generalization of this problem to massive
{\em hypergraphs}. Just as graphs model pairwise interactions between entities
in a network, hypergraphs model higher arity interactions. For instance, in an
academic collaboration network with researchers being the vertices, it would
be natural to model coauthorships on research papers and articles using {\em
hyperedges}, each of which can be incident to more than two vertices. Just as
we may use triangle counts to study clustering behaviors in graphs or even in
different portions of a single graph, we may analyze higher-order clustering
behaviors in the $3$-uniform hypergraph $H$ formed by all three-way
coauthorships by counting {\em $3$-simplices} in $H$. A $3$-simplex on four
vertices $\{u,v,w,x\}$ is the structure formed by the hyperedges $\{uvw, uvx,
uwx, vwx\}$: it is the natural $3$-dimensional analog of a triangle in an
ordinary graph. 

\subsection{Our Results} \label{sec:results}

We design several algorithms for space-efficiently estimating the number of
$k$-simplices in a $k$-uniform hypergraph $H$ that is given as a stream of
hyperedges.  A {\em $k$-simplex} is a complete $k$-uniform hypergraph on $k+1$
vertices.  The special case $k=2$ is the triangle counting problem which, as
noted above, is intensely investigated. Indeed, even in this setting of
streaming algorithms, triangle counting is highly studied, with new
algorithmic techniques being developed as recently as early
2021~\cite{JayaramK21}.  There is also a body of work on generalizing these
results to the problem of estimating the number of occurrences of patterns
(a.k.a.~{\em motifs}) more complicated than triangles, e.g., fixed-size
cliques and
cycles~\cite{ManjunathMPS11,KaneMSS12,PavanTTW13,BeraC17,KallaugherMPV19,Vorotnikova20}.
Our work adds to this literature by generalizing in a different direction: we
generalize the class of {\em inputs} from graphs to hypergraphs and focus on
counting the simplest nontrivial symmetric motifs, i.e., $k$-simplices. 

Our algorithms provide optimal space bounds (up to log factors) in certain
parameter regimes; we prove this optimality by giving a set of matching lower
bounds. In certain other parameter regimes, there remains a gap between our
best upper bounds and lower bounds, which immediately provides a goal for
future work on this problem.  Below are informal statements of our major
algorithmic results. 

\begin{theorem}[Upper bounds; informal] \label{thm:ubs-informal}
  Let $H$ be an $n$-vertex $k$-uniform hypergraph, presented as a stream of $m$
  hyperedges, that is promised to contain $T$ or more $k$-simplices. Suppose
  that each hyperedge is contained in at most $\Delta_E$ such simplices and
  each vertex is contained in at most $\Delta_V$ of them. Then, there are
  algorithms for $(1\pm\eps)$-estimating $T_k(H)$, the number of $k$-simplices 
  in $H$, with the following guarantees:
  \begin{thmparts}
    \item \label{part:ub-abundant} a $4$-pass algorithm using $\tO(m^{1+1/k}/T)$ space;
    \item \label{part:ub-meager} a $2$-pass algorithm using $\tO(m/T^{2/(k+1)})$ space, 
    provided $k \ge 3$; and
    \item \label{part:ub-1pass} a $1$-pass algorithm using 
    $\tO(m(\Delta_E + \Delta_V^{1-1/k})/T)$ space.
  \end{thmparts}
  Each of these algorithms is randomized and fails with probability at most
  $\delta$. 
    \footnote{Throughout this paper, the notation $\tO(\cdot)$ hides factors 
    of $O(\eps^{-2} \log\delta^{-1} \polylog n)$ and, in the context of
    $k$-uniform hypergraphs, treats $k$ as a constant. Also, the notation
    ``log'' with an unspecified base means $\log_2$.}
\end{theorem}

Formal versions of these results appear as \Cref{thm:ub-abundant} in
\Cref{sec:ub-abundant}, \Cref{thm:ub-meager-best} in
\Cref{sec:ub-meager-best}, and \Cref{thm:ub-meager-1pass} in
\Cref{sec:ub-meager-1pass}, respectively. Along the way, we also obtain some
other algorithmic results---stated as \Cref{thm:ub-abundant-subopt},
\Cref{thm:ub-meager-worst}, and \Cref{thm:ub-meager-better}---that we include
to paint a more complete picture, even though the space complexities of those
algorithms are dominated by the algorithms behind \Cref{thm:ubs-informal}. In
\Cref{sec:techniques}, we give a high-level overview of the techniques behind
these algorithmic results. As we shall see, we take several ideas from
triangle-counting algorithms as inspiration, but the ``correct'' way to extend
these ideas to hypergraphs is far from obvious. Indeed, we shall see that some
of the more ``obvious'' extensions lead to the less-than-best algorithms
hinted at above.

Next, we informally state our lower bound results. These are not the main
technical contributions of this paper, but they play the important role of
clarifying where our algorithms are optimal and where there might be room for
improvement. As in \Cref{thm:ubs-informal}, we denote the number of
$k$-simplices in $H$ by $T_k(H)$.

\begin{theorem}[Lower bounds; informal] \label{thm:lbs-informal}
  Let $k, n, m, H, \Delta_E, \Delta_V$ be as above.  Suppose an algorithm
  makes $p$ passes over a stream of hyperedges of $H$, using at most $S$ bits
  of working memory, and distinguishes between the cases $T_k(H) = 0$ and
  $T_k(H) \ge T$ with probability at least $2/3$. Then the following lower
  bounds apply.
  \begin{thmparts}
    \item With $T = 1$ and $p = O(1)$, a sublinear-space solution is impossible: 
    we must have $S = \Omega(n^k)$.
    \item With $p = 1$, we must have $S = \Omega(m\Delta_E/T)$.
    \item With $p = O(1)$, we must have $S =
    \Omega(m^{1+1/k}/T)$, $S = \Omega(m/T^{1-1/k})$, and 
    $S = \Omega(m\Delta_V^{1/k}/T)$.
  \end{thmparts}
  Here, a bound of the form $\Omega(Q)$ should be interpreted as ruling
  out the existence of an algorithm that can guarantee a space bound of $o(Q)$.
  If, instead, the streaming algorithm distinguishes between the cases $T_k(H)
  < (1-\eps)T$ and $T_k(H) > (1+\eps)T$, then the following lower bound
  applies.
  \begin{thmparts}[resume]
    \item With $p = O(1)$, we must have $S = \Omega(\eps^{-2})$.
  \end{thmparts}
\end{theorem}
There is some definitional subtlety in the use of asymptotic notation in these
lower bound results, hinted at in the language above. The picture will become
clearer when we state these results formally, in \Cref{sec:lb}.

\subsection{Closely Related Previous Work} \label{sec:closely-related}

To the best of our knowledge, this is the first work to study the simplex
counting problem in hypergraphs in the setting just described (the work
of~\cite{KallaugherKP18} on counting general hypergraph patterns is not
closely related; see \Cref{sec:other-related}). We now summarize some
highlights of previous work on the triangle counting problem (in graphs), with
a focus on streaming algorithms, so as to provide context for our
contributions.

Suppose an $n$-vertex $m$-edge graph $G$ is given as a stream of edges and we
wish to estimate the number of triangles, $T_3(G)$. It is not hard to show
that, absent any promises on the structure of $G$, this problem requires
$\Omega(n^2)$ space, even with multiple passes, thus precluding a
sublinear-space solution. Therefore, all work in this area seeks bounds under
a promise that $T_3(G) \ge T$, for some nontrivial threshold $T$. Intuitively,
the larger this threshold, the easier the problem, so we expect the space
complexity to decrease. The earliest nontrivial streaming
solution~\cite{BarYossefKS02} reduced triangle counting to a combination of
$\ell_0$, $\ell_1$, and $\ell_2$ estimation and achieved $\tO((mn/T)^2)$ space
by using suitable linear sketches. Almost all algorithms developed since then
have instead used some sort of {\em sampling} to extract a small portion of
$G$, perform some computation on this sample, and then extrapolate to estimate
$T_3(G)$. 

Over the years, a number of different sampling strategies have been developed,
achieving different, sometimes incomparable, guarantees. Here is a whirlwind
tour through this landscape of strategies. One could sample an edge uniformly
at random (using reservoir sampling), then count common neighbors of its
endpoints~\cite{JowhariG05}; or sample an edge uniformly and sample a vertex
not incident to it~\cite{BuriolFLMS06}; or sample a subset of edges by
independently picking each with a carefully adjusted probability
$p$~\cite{KolountzakisMPT12,BravermanOV13}; or choose a random color for each
vertex and collect all monochromatic edges~\cite{PaghT12}; or sample a subset
of vertices at random and collect all edges incident to the
sample~\cite{BravermanOV13}.  One could collect two random subsets of vertices
at different sampling rates and further sample edges between the two
subsets~\cite{KallaugherP17}; or, as in a very recent algorithm, sample a
subset of vertices at rate $p$ and further sample edges incident to this
sample at rate $q$, for well-chosen $p$ and $q$~\cite{JayaramK21}. Notice that
in the just-mentioned algorithms, the sampling technique decides whether or
not to store an edge without regard to other edges that may be in store and is
thus not actively trying to ``grow'' a triangle around a sampled edge.  We
shall call such sampling strategies {\em oblivious}.

Besides the above oblivious sampling strategies, another set of works used
what we shall call {\em targeted sampling} strategies,%
\footnote{To be perfectly honest, the terms ``targeted sampling'' and
``oblivious sampling'' do not have precise technical definitions, but we hope
the conceptual distinction is helpful to the reader as it was helpful to us.}
where information previously stored about the stream guides what gets sampled
subsequently.  Here is another quick tour through these. One could sample {\em
wedges} (defined as length-$2$ paths) in the input graph, using a more
sophisticated reservoir sampling approach~\cite{JhaSP15}; or sample an edge
uniformly and then sample a second edge that touches the
first~\cite{PavanTTW13}; or sample a vertex with probability proportional to
its squared degree, then sample two neighbors of that
vertex~\cite{McGregorVV16}; or sample an edge uniformly and then sample a
neighbor of the lower-degree endpoint of that edge~\cite{BeraC17}. 

There are also a handful of algorithms that add further twists on top of the
sample-count-extrapolate framework.  The algorithm of~\cite{BulteauFKP16}
combines the vertex coloring idea of \cite{PaghT12} with $\ell_2$ estimation
sketches to obtain a solution that can handle {\em dynamic} graph streams,
where each stream update may either insert or delete an edge. The algorithm
of~\cite{CormodeJ17} combines multiple runs of \cite{PaghT12} with a {\em
heavy/light edge partitioning} technique: an edge is deemed ``heavy'' if it
participates in ``too many'' triangles. A key observation is that the variance
of an estimator constructed by oblivious sampling---which needs to be small in
order to guarantee good results in small space---can be bounded better if no
heavy edges are involved. On the other hand, triangles involving a heavy edge
are easier to pick up (because there are many of them!) by randomly sampling
vertices at a low rate. Thus, by carefully picking the threshold for
heaviness, one can combine an algorithm that counts all-light triangles
efficiently with one that counts heavy-edged triangles efficiently for a good
overall space bound. In this way, \cite{CormodeJ17} obtains a space bound of
$O(\eps^{-2.5} \log\delta^{-1} \polylog n \cdot m/\sqrt{T})$, while
\cite{McGregorVV16} provides a tight dependence on $\eps$ (i.e., $\eps^{-2}$)
by achieving $\tO(m/\sqrt{T})$ space.

Separately, the aforementioned targeted sampling strategies of
\cite{McGregorVV16} and \cite{BeraC17} provide $4$-pass algorithms for
estimating $T_3(G)$ using space $\tO(m^{3/2}/T)$. When $T$ is large
enough---specifically, $T = \Omega(m)$---this space bound is better than the
$\tO(m/\sqrt{T})$ bound obtained via heavy/light edge partitioning. By
picking the better of the two algorithms, one obtains a space bound of
$\tO(\min\{m^{3/2}/T, m/\sqrt{T}\})$. Both portions of this bound are tight,
thanks to lower bounds of $\Omega(m^{3/2}/T)$ and $\Omega(m/\sqrt{T})$ that
follow by reducing from the \textsc{set-disjointness} communication
problem~\cite{CormodeJ17,BeraC17}.

The algorithm of \cite{JayaramK21} is optimal in a different sense: it runs in
a single pass and $\tO((m/T)(\Delta_E + \sqrt{\Delta_V}))$ space, where
$\Delta_E$ and $\Delta_V$ are as defined in \Cref{thm:ubs-informal}. Each term
in this bound is tight, since the aforementioned reduction also implies an
$\Omega(m\sqrt{\Delta_V}/T)$ lower bound and a different reduction from the
\textsc{index} communication problem implies an $\Omega(m\Delta_E/T)$ bound
for $1$-pass algorithms~\cite{BravermanOV13}. Note, however, that this result
is incomparable to the multi-pass upper bounds noted above. It must be so: a
lower bound of $\Omega(m^3/T^2)$ holds for $1$-pass algorithms~\cite{BeraC17}.

\subsection{Other Related Work} \label{sec:other-related}

This work is focused on streams that simply {\em list} the input hypergraph's
hyperedges. This is sometimes called the {\em insert-only} streaming model, in
contrast to the {\em dynamic} or {\em turnstile} streaming model where the
stream describes a sequence of (hyper)edge insertions or deletions. A small
subset of works mentioned in \Cref{sec:closely-related} do provide results in a
turnstile model. Besides these, there is the recent seminal work
\cite{KallaugherKP18} that fully settles the complexity of triangle counting
in turnstile streams for {\em constant-degree} graphs. This work also
considers the very general problem of counting copies of an arbitrary
fixed-size hypergraph motif $M$ inside a large input hypergraph $H$, again in
a turnstile setting. Because of the way their upper bound results depend on the
structure of $M$, they cannot obtain sublinear-space solutions for counting
$k$-simplices without a strong constant-degree assumption on the input $H$.

A handful of works on triangle counting consider adjacency list
streams~\cite{BuriolFLMS06,KolountzakisMPT12,McGregorVV16}, where the input
stream provides all edges incident to each vertex contiguously.
This setting can somewhat simplify algorithm design,
though the basic framework is still sample-count-extrapolate. We do not
consider adjacency list streams in this work.

There are important related algorithms that predate the now-vast literature on
streaming algorithms. In particular, \cite{AlonYZ97} gives the current best
run-time for exact triangle counting in the RAM model and \cite{ChibaN85} gives
time-efficient algorithms for listing all triangles (and more general motifs).
A version of the heavy/light partitioning idea appears in these early works.

More recently, a handful of works \cite{EdenLRS15,AssadiKK19,EdenRS20} have
designed {\em sublinear-time} algorithms for approximately counting triangles
and other motifs given query access to a large input graph. Triangle
detection, listing, and counting have connections to other important problems
in the area of fine-grained complexity \cite{GronlundP18,Williams21}. Triangle
counting has also been studied in distributed, parallel, and high-performance
computing models \cite{SuriV11,PaghT12,SeshadhriPK14,GuiZYLJ19}.

\section{Our Algorithmic Techniques} \label{sec:techniques}

In this section, we shall describe the algorithms designed in this work at a
high level. Our goal is to give an overview of our techniques so as to clarify
two things: (a)~how our various algorithms relate to one another, and (b)~how
we build upon several of the techniques described in
\Cref{sec:closely-related} and what novelty we add.

Recall that our input is a $k$-uniform hypergraph $H$ that has $m$ edges; the
motif we're counting is a $k$-simplex (which involves $k+1$ vertices); and all
our algorithms are given a parameter $T$, which is a promised lower bound on
$T_k(H)$, the number of $k$-simplices.  Our algorithms can be divided into two
families. Algorithms in the first family use {\em targeted sampling} in the
sense described in \Cref{sec:closely-related} and provide space guarantees of
the form $\tO(m^\lambda/T)$, for some $\lambda > 1$.  Algorithms in the second
family use various {\em oblivious sampling} strategies, again in the sense
described in \Cref{sec:closely-related}, and their space guarantees are
typically of the form $\tO(m/T^\lambda)$, for some $\lambda > 1$. Note that,
given two specific algorithms $\cA_1, \cA_2$ from the first and second
families respectively, there is a threshold $\tau$ such that the more
space-efficient of the two algorithms is $\cA_1$, when $T > \tau$, and
$\cA_2$, otherwise. Thus, algorithms from the first family are good when
$k$-simplices are ``abundant'' in $H$, whereas those from the second family
are good when $k$-simplices are ``meager'' in $H$.  For a concrete example,
take $k=3$: our results give space bounds of $\tO(m^{4/3}/T)$ and
$\tO(m/T^{1/2})$; the former wins when $T > m^{2/3}$ and the latter wins
otherwise.

\subsection{The ``Abundant'' Case: Targeted Sampling} \label{sec:tech-abundant}

Our first family of algorithms uses targeted sampling along the lines of
\cite{BeraC17}: that is, we sample a hyperedge $e$ uniformly at random and
then sample a ``neighboring'' vertex $v$ by considering edges that interact
with $e$ in a good way.  Having made these choices, the $k+1$ vertices in $e
\cup \{v\}$ either do or do not define a simplex; we detect which is the case
and prepare our {\em basic estimator} accordingly. It is not hard to make this
estimator unbiased, by using appropriate scaling.

The nontrivial portion of these algorithms is controlling the variance
$\sigma^2$ of the basic estimator. The final algorithm, which returns a
$(1\pm\eps)$-approximation to $T_k(H)$ with high probability, is obtained by
combining several independent copies of the basic estimator using a standard
median-of-means technique (\Cref{lem:median-of-means}).
The eventual space complexity is proportional to the number of independent
copies, which needs be proportional to $\sigma^2$.  What might make $\sigma^2$
high? If we ``detect'' a $k$-simplex $\Xi$ upon sampling any one of its
hyperedges $e$ and picking up the sole remaining vertex $v$, then we may run
into trouble when $H$ contains many simplices that share a common hyperedge
$e^*$: our estimator is too drastically affected by whether or not the initial
random sampling picked $e^*$. 

To remedy this, we modify our basic estimator so that it detects the simplex
formed by $e$ and $v$ (if it exists) only when $e$ has low influence on this
detection. In the graph case ($k=2$), the algorithm of \cite{BeraC17} chooses
$v$ from among the neighbors of the {\em lower-degree} endpoint of $e$ and
considers a triangle detected only when $v$ is its highest-degree vertex:
these choices are crucial for the combinatorial arguments that prove their
variance bound.

\paragraph{A Suboptimal Algorithm.}
We start with a natural extension of the above idea to hypergraphs. Having
chosen $e$, we then choose $v$ from the neighborhood of the minimum degree
vertex inside $e$. We consider a simplex at the vertices in $e \cup \{v\}$ to
be detected iff $v$ is its highest-degree vertex. This sets up a basic
estimator and then a final estimator as described above. Algorithmically, we
implement each basic estimator by using one pass to pick $e$ and additional
passes to pick $v$ and test for the presence of a simplex. Analyzing the
resulting algorithm largely comes down to upper-bounding the variance of the
basic estimator.

This variance bound is not straightforward. After some initial calculations,
we arrive at an expression involving the sum
\begin{align} 
  s_1 := \sum_{e \in E(H)} \min\{\deg(v):\, v \in e\} \,,
  \label{eq:sum-min-deg}
\end{align}
for which we need a good upper bound. In the case of a graph $G$, such a bound
is provided by a classic result of \cite{ChibaN85}, which states that
$\sum_{\{u,v\} \in E(G)} \min\{\deg(u),\deg(v)\} = O(m^{3/2})$. The
combinatorial argument that proves this ultimately involves {\em orienting}
each edge. We need to find the right analog of this notion for hypergraphs and
this is where the complication lies.  Our solution is to generalize the graph
theoretic notion of {\em arboricity} to what we call {\em hyperarboricity},
based on decomposing the set $E(H)$ of hyperedges into {\em hyperforests} as
defined by \cite{FrankKK03}. We then prove a general upper bound on the
hyperarboricity of an $m$-edge hypergraph in terms of $m$ and this in turn
proves that $s_1 = O(m^{2-1/k})$.  From here, it is not hard to obtain a space
bound of $\tO(m^{2-1/k}/T)$ for estimating $T_k(H)$. 

Though this space bound is suboptimal, we feel that the algorithm and its
analysis are interesting in their own right and the ideas could be instructive
for future work. Therefore, we present the details towards the end of the
paper, in \Cref{sec:subopt}.

\paragraph{An Optimal Algorithm.}
To obtain the optimal bound of $\tO(m^{1+1/k}/T)$ mentioned in
\Cref{thm:ubs-informal}, we change the basic estimator as follows. We again
start by choosing $e$ uniformly at random, but for choosing the additional
vertex $v$, we do something more complicated: we consider {\em joint
neighborhoods} and {\em codegrees} of sets of vertices. 

For ease of high-level exposition, let us consider the case $k=3$ first. Given
two vertices $x$ and $y$, the joint neighborhood of $\{x,y\}$ is the set
$\{u:\, \{x,y,u\} \in E(H)\}$ and the codegree $\codeg(\{x,y\})$ is the number
of hyperedges containing both $x$ and $y$. Suppose that the initially picked
hyperedge is $e = \{x,y,z\}$, with $x$ being the minimum-degree vertex among
these three and $y$ being such that $\codeg(\{x,y\}) \le \codeg(\{x,z\})$. We
then sample $v$ from the joint neighborhood of $\{x,y\}$. Having done so, if
there is in fact a $3$-simplex on the vertices $\{x,y,z,v\}$, we ``detect'' it
if and only if the following conditions hold: (a)~$x$ has minimum degree among
$x,y,z$, and $v$; (b)~$\{x,y\}$ has minimum codegree among $\{x,y\}$,
$\{x,z\}$, and $\{x,v\}$; and (c)~$\{x,z\}$ has minimum codegree among
$\{x,z\}$ and $\{x,v\}$. 

For general $k\ge 3$, we define a careful ordering of the vertices inside $e$
using codegrees of successively larger vertex sets.  A single streaming pass
suffices to determine this ordering within $e$.  We sample $v$ from the joint
neighborhood of the first $k-1$ vertices in this ordering. Finally, we detect
a simplex at $\Xi := e \cup \{v\}$, if one exists, only if a similar iterative
procedure that repeatedly picks the ``smallest'' remaining vertex from $\Xi$
singles out $v$ as the only vertex not picked. The full details and analysis
appear in \Cref{sec:ub-abundant} after the necessary definitions in
\Cref{sec:prelim}.

As before, the space complexity analysis hinges on bounding the variance of
the basic estimator. After some algebra, this boils down to giving a good
bound for a combinatorial sum that is a more complicated version
of~\eqref{eq:sum-min-deg}. In the new sum, the term corresponding to a
particular $e \in E(H)$ is the codegree of the first $k-1$ vertices in the
above ordering within $e$. We prove that this sum is bounded by
$O(m^{1+1/k})$, leading to our claimed optimal space bound.

\subsection{The ``Meager'' Case: Oblivious Sampling} \label{sec:tech-meager}

We now outline the ideas involved in our second family of algorithms, which
use various {\em oblivious} sampling strategies. This means that the sampling
portion of one of these algorithms picks (and stores) a hyperedge based on
conditions tested only on that specific hyperedge, without regard to what else
has been stored.  The resulting space bounds are typically of the form
$\tO(m/T^\lambda)$.

The most basic of these is a direct generalization of one of the algorithms in
\cite{McGregorVV16}, which picks each edge independently with probability $p$
and applies heavy/light edge partitioning on top of this. Let us explain this
suboptimal algorithm in slightly more detail, so as to introduce a general
framework for analyzing this family of algorithms. In fact, this framework is
an important conceptual contribution of this work.

\paragraph{A First Attempt and a Framework for Analysis.}
We collect a sample $Q$ by picking each incoming hyperedge with probability $p$,
independently. Define a {\em hyperwedge} to be a simplex minus one of its
hyperedges (generalizing the notion of an {\em wedge} in a graph, which is a a
triangle minus one edge). For our algorithm, we detect a simplex $\Xi$ if all
$k$ edges of one of its hyperwedges are chosen in $Q$: the actual detection
happens in a subsequent pass, when we see the sole remaining hyperedge of $\Xi$ in
the stream. It is easy to show that the number of detected simplices divided
by $(k+1)p^k$ is an unbiased estimator of $T_k(H)$. It will now be helpful to
think of the algorithm as counting the hyperwedges defined by the
$k$-simplices in $H$: this count is exactly $T' := (k+1)T_k(H)$. A particular
hyperwedge is detected iff all its edges are sampled.

It is not hard to bound the variance of our unbiased estimator by an
expression whose important term is proportional to 
\begin{align}
  s_2 := \sum_{\ontop{1 \le i < j \le T'}{i \sim j}}
    \Pr[\text{both } W_i \text{ and } W_j \text{ are detected}] \,.
\end{align}
where $W_1, \ldots, W_{T'}$ enumerates all the copies within $H$ of the motif
being counted (in this case, hyperwedges) and ``\,$i \sim j$\,'' means that
the events ``$W_i$ is detected'' and ``$W_j$ is detected'' are correlated. At
this point, we observe that our sampling procedure satisfies the following 
\begin{quote}
  \textsc{Separation Property:~} If $W_i$ and $W_j$ have no hyperedge in common, 
  then ``$W_i$ is detected'' and ``$W_j$ is detected'' are independent events.
\end{quote}

Depending on the sampling algorithm in question, let the parameters $\alpha$ and $\beta$ be such that
\begin{align}
  \forall\, i &:~ \Pr[W_i \text{ is detected}] = p^\alpha \,, 
    \label{eq:alpha-def-informal} \\
  \forall\, i, j \text{ with } i \sim j &:~
    \Pr[W_i, W_j \text{ are both detected}] \le p^\beta \,, 
    \label{eq:beta-def-informal}
\end{align}
For the algorithm we are currently describing, $\alpha = k$ and $\beta = 2k-1$.  Using the
separation property and some algebra, we can bound $s_2$ by a term
proportional to $\Delta_E T_k(H) p^\beta$, where $\Delta_E$ is the maximum
number of simplices that may share a single hyperedge. Then, using Chebyshev's
inequality, we can bound the probability that our unbiased estimator is not a
$(1\pm\eps)$-approximation to $T_k(H)$ by
\begin{align}
  \Perr \le \frac{1}{\eps^2 p^\alpha T_k(H)} +
    \frac{(k+1)\Delta_E}{2 \eps^2 p^{2\alpha - \beta} T_k(H)} \,.
    \label{eq:estim-deviates}
\end{align}
The choice $p = \Theta(T^{-1/\alpha})$ makes the first term small (we'll ignore
the dependence on $\eps$ in this outline) and reduces the second term to
$O(\Delta_E p^{\beta - \alpha}) = O(\Delta_E / T^\theta)$, where $\theta =
\beta/\alpha - 1$. We cannot claim that this term is small in general.  The
solution, as in \cite{CormodeJ17, McGregorVV16}, is to apply this sampling
technique to a sub-hypergraph of $H$ where $\Delta_E = O(T^\theta)$ is
guaranteed and handle the rest of $H$ in some other way.

\paragraph{Heavy/Light Edge Partitioning.}
Define a hyperedge $e \in E(H)$ to be {\em heavy} if it participates in at
least $T^\theta$ simplices and {\em light} otherwise. Further, define a
simplex occurring in $H$ to be heavy if it at least one of its hyperedges is
heavy; define it to be light otherwise. Assume, for a moment, that as each
hyperedge arrives in the stream, an {\em oracle} tells us whether it is heavy or
light. Using such an oracle, we can filter out all heavy edges and apply the
above sampling scheme to the sub-hypergraph that remains, thereby obtaining a
good estimate for the number of light simplices.

As for the rest, the number of simplices containing a particular heavy
hyperedge $e$ can be estimated very well by sampling the {\em vertices} of $H$
at a rate of $\tO(T^{-\theta})$ and counting how many sampled vertices
complete a simplex with $e$. Algorithmically, this amounts to collecting the
set $S$ of all edges incident to sampled vertices in one streaming pass and
checking for the relevant simplices in a subsequent pass. We can therefore
estimate the number of heavy simplices; a bit of care is needed to account for
multiplicities when a simplex contains several heavy edges.

It remains to explain how the oracle we need can be implemented. A Chernoff
bound argument shows that $e$ is heavy iff it completes a noticeable number of
simplices using edges from $S$, so this same sample $S$ can be repurposed to
implement the oracle in the second pass. The algorithm is complete.

\paragraph{Space Bound.}
The space usage of the overall algorithm is dominated by the space required to
store the sample $Q$, for counting the light simplices, and the sample $S$,
for implementing the oracle and counting the heavy simplices. If the sampling
scheme stores a hyperedge in $Q$ with probability $p^\gamma$, for some
parameter $\gamma$ (note that, for the simple algorithm we are currently
considering, $\gamma = 1$), then, in expectation, $|Q| = p^\gamma m 
= \Theta(m / T^{\gamma/\alpha})$ whereas $|S| = \Theta(T^{-\theta} m) =
\Theta(m / T^{\beta/\alpha - 1})$. Thus, we obtain a space bound of
$\tO(m/T^\lambda)$ for
\[
  \lambda = \min\left\{ \frac{\gamma}{\alpha},\, \frac{\beta}{\alpha} - 1 \right\} \,.
\]
For the present algorithm, $\alpha = k, \beta = 2k-1$, and $\gamma = 1$,
leading to a space bound of $\tO(m/T^{1/k})$. Notice how this generalizes the
$\tO(m/\sqrt{T})$ bound for triangle-counting in graphs~\cite{McGregorVV16}.
Notice, also, that this space bound is already sublinear and, in a certain
regime of $T$ values, it beats $O(m^{1+1/k}/T)$, which we obtained by targeted
sampling. 

\paragraph{First Improvement: Sampling by Coloring.}
We can improve on the above space bound by using a more sophisticated sampling
strategy. Recall the \cite{PaghT12} strategy of coloring vertices and picking
up monochromatic edges. We could do the same here: use a coloring function
$\col$ to assign each vertex a random color from $\{1,\ldots,N\}$, where $N =
1/p$ (assume, for simplicity, that this is an integer) and store a hyperedge
iff it becomes monochromatic, an event that has probability $p^{k-1}$. Thus,
the set $Q$ of all stored hyperedges has expected size $p^\gamma m$ with
$\gamma = k-1$.

It is tempting to apply the above analysis framework to this algorithm, using
$\alpha = k$ and $\beta = k+1$---these values {\em do} satisfy
\cref{eq:alpha-def-informal,eq:beta-def-informal} with $W_i$ being the $i$th
copy of a $k$-simplex in $H$---to arrive at $\lambda = \min\{1-1/k, 1/k\} =
1/k$.  However, this would be wrong, because this sampling scheme breaks the
separation property defined above: two $k$-simplices within $H$ are correlated
as soon as they have two vertices in common. The analysis only works when
$k=2$ (the graph case), when we do obtain $\lambda = 1/2$ and a space bound of
$\tO(m/\sqrt{T})$.

To remedy the situation, we introduce our first novel twist: we apply the
coloring idea not to vertices, but to {\em subsets of vertices}.
Specifically, we assign one of $N = 1/p$ colors to each $(k-1)$-sized subset
of vertices and we collect a hyperedge $e$ into $Q$ iff all subsets within $e$
receive the same color. This removes the unwanted correlations between
simplices that don't share a common hyperedge. After going through the above
analysis framework, we arrive at $\lambda = 2/(k+2)$, for a space bound of
$\tO(m/T^{2/(k+2)})$.

\paragraph{Second Improvement: Working With a Shadow Hypergraph.}
To obtain the even better bound of $\tO(m/T^{2/(k+1)})$ mentioned in
\Cref{thm:ubs-informal}, we introduce a further twist, one that is specific to
hypergraphs (i.e., it requires hyperedge size $k \ge 3$). We perform the
vertex-subset coloring on what we call the {\em shadow hypergraph} of the
input $H$: this shadow is a $(k-1)$-uniform hypergraph $H'$ obtained by removing
the lowest-numbered vertex from each hyperedge of $H$ and doing some
bookkeeping to record which vertex got removed. This bookkeeping creates a more
complicated set $V'$ of shadow vertices, so that the set of shadow edges $E'$
is in bijective correspondence with the original edge set $E$.

As we shall see in \Cref{sec:ub-meager}, such shadow edges contain most of the
information we need to count $k$-simplices in $H$ and we obtain a space
efficient algorithm by sampling shadow edges based on a coloring of
$(k-2)$-sized subsets of shadow vertices. Notice how this requires $k \ge 3$.
All further details are best explained after the construction of $H'$ has been
formally given. Once we spell out these details and go through our usual
analysis framework, we obtain $\lambda = 2/(k+1)$, leading to the claimed
space complexity.

\paragraph{A One-Pass Algorithm.}
Finally, we turn to the $1$-pass bound claimed in \Cref{thm:ubs-informal},
where the space bound is given in terms of structural parameters $\Delta_E$
and $\Delta_V$, in addition to the usual $m$ and $T$.  This is based on
generalizing the recent algorithm of \cite{JayaramK21}, i.e., sampling
vertices at a certain rate and then sampling incident hyperedges at a certain
other rate. The corresponding analysis is simpler, since we do not use the
heavy/light partitioning technique. In fact, it is a fairly direct
generalization of the analysis in \cite{JayaramK21}.

\section{Preliminaries} \label{sec:prelim}

We now define our key terminology, set up notation, and establish some basic
facts that we shall refer to in the algorithms and analyses to come.

\begin{definition}[Hypergraph, degrees, neighborhoods] \label{def:hypergraph}
  A {\em hypergraph} is a pair $H = (V,E)$ where $V$ is a nonempty finite set
  of {\em vertices} and $E \subseteq 2^V$ is a set of {\em hyperedges}. In
  certain contexts, this structure is instead called a {\em set system} over
  $V$. If $|e| = k \geq 1$ for all $e \in E$, then $H$ is said to be {\em $k$-uniform}. 
  For simplicity, we shall often shorten ``$k$-uniform hypergraph'' to {\em
  $k$-graph} and ``hyperedge'' to ``edge.''

  For each $S \subseteq V$ with $|S| < k$, the {\em joint degree} or
  {\em codegree} $\codeg(S)$ of $S$ is the number of edges that strictly
  extend $S$. The {\em neighborhood} $\Nhd(S)$ of $S$ is the set of non-$S$ vertices
  that share an edge with $S$. Formally,
  \begin{align} 
    \codeg(S) &:= |\{e \in E:\, e \supsetneq S\}| \,;
      \label{eq:codeg} \\
    \Nhd(S) &:= \{v \in V:\, v \notin S \text{ and } 
      \exists\,e\in {E} \text{ such that } e \supseteq S\cup\{v\}\} \,.
      \label{eq:nhd}
  \end{align}
  For a singleton set $S = \{u\}$, we define $\deg(u) := \codeg(\{u\})$ and
  $\Nhd(u) := \Nhd(\{u\})$. Further, given $S \subseteq V$, we define the 
  {\em $S$-relative degrees} of vertices $x \in V \setm S$ by
  \begin{align}
    \reldeg{S}{x} := \codeg(S \cup \{x\}) \,.
      \label{eq:reldeg}
  \end{align}
  Note that this definition is meaningful only when $|S| \le k-2$. Note, also,
  that $\deg(x) = \reldeg{\varnothing}{x}$.
\end{definition}

Throughout the paper, hypergraphs will be $k$-uniform unless qualified
otherwise. We shall consistently use the notation $H = (V,E)$ for a generic
$k$-graph and define $n := |V|$ and $m := |E|$. We shall assume that each
vertex $v \in V$ has a unique ID, denoted $\ID(v)$, which is an integer 
in the range $[n] := \{1,\ldots,n\}$.

In general, there isn't an
equation relating $|\Nhd(S)|$ to $\codeg(S)$. However, all $k$-graphs satisfy the following useful lemmas.

\begin{lemma}[Degree-vs-Neighborhood] \label{lem:deg-nhd}
  For $S \subseteq V$ with $|S| < k$,
  $\codeg(S) \le |\Nhd(S)| \le (k-|S|)\codeg(S)$.
\end{lemma}
\begin{prf}
  Each edge that extends $S$ adds at least one and at most $k-|S|$ 
  new neighbors to the sum.
\end{prf}

\begin{lemma}[Generalized Handshake] \label{lem:handshake}
  For $1 \le r < k$, we have
  $\sum_{S \subseteq V:\, |S|=r} \codeg(S) = \binom{k}{r} m = O(m)$.
\end{lemma}
\begin{prf}
  The sum counts each edge exactly $\binom{k}{r}$ times.
\end{prf}

\begin{definition}[Simplices]
  A {\em $k$-simplex} is a $k$-graph on $k+1$ vertices such that all possible
  $k$-sized edges are present. If $H = (V,E)$ is a $k$-graph and $X \subseteq
  V$ with $|X| = k+1$, we say that $H$ has a simplex at $X$ if the induced subhypergraph $H[X]
  := (X, \{e \in E: e \subseteq X\})$ is a simplex. Abusing notation, we also
  use $X$ to denote this simplex.

  We use $\cT_k(H)$ to denote the set of all $k$-simplices in $H$ and $T_k(H)
  := |\cT_k(H)|$ to denote the number of such simplices. We define
  $\Delta_E = \Delta_E(H)$ and $\Delta_V = \Delta_V(H)$ as follows:
  \begin{align}
    \Delta_E &:= \max_{e \in E} |\{X \in \cT_k(H):\, X \text{ contains edge } e\}| \,; 
      \label{eq:delta_e} \\
    \Delta_V &:= \max_{v \in V} |\{X \in \cT_k(H):\, X \text{ contains vertex } v\}| \,.
      \label{eq:delta_v}
  \end{align}
\end{definition}

Drawing inspiration from the {\em link} operation for abstract and simplicial
complexes in topology, we define a couple of operations that derive $(k-1)$-graphs
from $k$-graphs. 

\begin{definition}[Neighborhood and shadow hypergraphs] \label{def:nhd-shadow}
  Let $k \ge 3$ and let $H = (V,E)$ be a $k$-graph. For each $u \in V$, the
  {\em neighborhood hypergraph of $H$ at $u$} is the $(k-1)$-graph
  $\nhyp{H}{u} = (\Nhd(u), E_u)$, where
  \[
    E_u := \{\{x_1,\ldots, x_{k-1}\}:\, \{u, x_1, \ldots, x_{k-1}\} \in E\} 
    = \{e \setm u:\, e \in E \text{ and } e \ni u\} \,.
  \]
  The {\em shadow hypergraph} of $H$ is the $(k-1)$-graph $H' = (V',E')
  = \left( \bigcup_{u\in V} V'_u, \, \bigcup_{u \in V} E'_u \right)$, where
  \begin{align*}
      V'_u &:= \big\{ \flavor{u}{x}:\, x\in V \big\}
        \text{ is a copy of $V$ ``flavored'' with $u$, and} \\
      E'_u &:= \big\{\big\{ \flavor{u}{x_1}, \ldots, \flavor{u}{x_{k-1}} \big\}:\,
        \{u,x_1,\ldots,x_{k-1}\} \in E \text{ and }
        \ID(u) \leq \ID(x_i)~ \forall\, i\in[k-1] \big\} \,.
  \end{align*}
  Note that $E'_u$ is subtly different from $E_u$ in that it is induced by
  hyperedges incident on $u$ in which $u$ is the minimum-ID vertex. Observe
  that, as a result, $|E'| = |E|$.
\end{definition}
We use neighborhood hypergraphs as a tool to analyze our first major
algorithm, in \Cref{sec:ub-abundant}. We use shadow hypergraphs in a crucial
way to obtain our second major algorithm, in \Cref{sec:ub-meager-best}. 

In the literature on triangle counting in graphs, the structure formed by two
edges sharing a common vertex is called a {\em wedge}. It will be useful to 
define an analog of the notion for hypergraphs.
\begin{definition} \label{def:hyperwedge}
  A {\em $k$-hyperwedge} is the hypergraph obtained by deleting one hyperedge
  from a $k$-simplex. Thus, if its vertex set is $X$, then $|X| = k+1$ and
  it has $k$ hyperedges with exactly one common vertex, which we call the {\em apex}
  of the hyperwedge. The only $k$-sized subset of $X$ that is {\em not} present in
  the hyperwedge is called the {\em base} of the hyperwedge. Adding this base as
  a new hyperedge produces a $k$-simplex.
\end{definition}

Fix a $k$-graph $H$ and a $k$-simplex $X$ in $H$. Observe that $X$
contains $k+1$ distinct hyperwedges; each such hyperwedge $W$ corresponds to a 
$(k-1)$-simplex in the neighborhood hypergraph of $H$ at the apex of $W$.
On the other hand, $X$ corresponds to exactly one $(k-1)$-simplex in the
shadow hypergraph $H'$, namely the simplex at 
$\{\flavor{z}{u_1}, \ldots, \flavor{z}{u_k}\}$, where $z$ is the
minimum-ID vertex in $X$ and $\{u_1,\ldots,u_k\} = X \setm \{z\}$.

\bigskip
The first major algorithmic result we shall give estimates $T_k(H)$ in 
space $\tO(m^{1+1/k}/T)$. To appreciate this form of the bound, it is worth
studying the following simple result that upper bounds the maximum possible
value of $T_k(H)$. It is a direct analogue of the well-known graph theoretic
fact that a graph with $m$ edges has at most $O(m^{3/2})$ triangles. The proof
will serve as a useful warm-up for the work to come.

\begin{proposition}[Upper bound on number of simplices] \label{prop:max-simplices}
  For every $k$-graph $H = (V,E)$ with $m$ edges, we have the bound $T_k(H) =
  O\left(m^{1+1/k}\right)$.
\end{proposition}
\begin{proof}
  We use induction on $k$. The $k=2$ case is the aforementioned folklore
  result for triangles in graphs.

  Consider $k \ge 3$. For $v \in V$, call $v$ {\em heavy} if $\deg(v) >
  m^{(k-1)/k}$ and {\em light} otherwise. By the handshake lemma, there are
  $O(m^{1/k})$ heavy vertices, each of which can be in at most $m$ simplices.
  Therefore, the number of simplices that contain a heavy vertex is at most
  $O(m^{1+1/k})$.

  Now, let $\Tklight(H)$ be the number of simplices with all their vertices
  light. Each simplex containing a particular light vertex $v$ must, as
  observed above, correspond to a $(k-1)$-simplex in the neighborhood
  hypergraph $\nhyp{H}{v}$. Thus, by the induction hypothesis,
  \begin{align*}
    \Tklight(H) &= O\left(\sum_{v \text{ light}}\deg(v)^{k/(k-1)}\right) \\
    &= O\left(\sum_{v \text{ light}} \deg(v)^{1/(k-1)} \cdot \deg(v)\right) \\
    &= O\left(m^{1/k} \sum_{v \text{ light}}\deg(v)\right)
    = O\big(m^{1+1/k}\big) \,. \qedhere
  \end{align*}
\end{proof}
The above bound is tight. A complete $n$-vertex $k$-graph has $T_k(H) =
\Theta(n^{k+1})$ and $m = \Theta(n^{k})$.

\bigskip
A random variable $\widetilde{Q}$ is said to be an {\em
$(\eps,\delta)$-estimate} of a statistic $Q$ if $\Pr[\widetilde{Q} \in
(1\pm\eps) Q] \ge 1-\delta$. The following boosting lemma from the theory of
randomized algorithms is standard and we shall often use it.
\begin{lemma}[Median-of-Means improvement] \label{lem:median-of-means}
  If $\hat{Q}$ is an unbiased estimator of some statistic, then
  one can obtain an $(\eps,\delta)$-estimate of that statistic by
  suitably combining
  \[
    K := \frac{C}{\eps^2} \ln\frac{2}{\delta} \cdot
      \frac{\Var[\hat{Q}]}{\EE[\hat{Q}]^2}
  \]
  independent samples of $\hat{Q}$, where $C$ is a universal constant. \qed
\end{lemma}

\section{An Optimal Algorithm Based on Targeted Sampling} \label{sec:ub-abundant}

In this section, we present the details of the algorithm behind
\cref{part:ub-abundant} in \Cref{thm:ubs-informal}, with a formal description
of the algorithm (as \Cref{alg:abundant}) followed by its analysis. This
algorithm is based on ``targeted sampling,'' as outlined in
\Cref{sec:techniques}. It runs in $\tO(m^{1+1/k}/T)$ space, which is optimal
in view of \Cref{thm:lb-main}. It is the method of choice for the ``abundant''
case, when $T$ is large enough: specifically, in view of the guarantees of
\Cref{alg:meager-best} to follow (see \Cref{thm:ub-meager-best}), ``large
enough'' should be defined as $T \ge m^{(k+1)/(k^2-k)}$.

As outlined in \Cref{sec:tech-abundant}, the algorithm presented in this
section is in fact the second (and the less straightforward) of two
targeted-sampling algorithms. The first of these algorithms gives a suboptimal
space bound of $\tO(m^{2-1/k})$: it is easier to describe but its space
guarantee is far from obvious to prove, requiring new methods for analyzing
the structure of hypergraphs. Since this analysis is interesting on its own,
we do give a full exposition and analysis of this suboptimal algorithm, but
defer it to \Cref{sec:subopt}.

\subsection{The Algorithm} \label{sec:abundant-alg}

The basic setup is as in the algorithm of \cite{BeraC17} for counting odd
cycles in graphs (and an analogous algorithm of \cite{McGregorVV16}). We use
one pass to pick an edge $e \in E$ uniformly at random and, in subsequent
passes, use further sampling to estimate the number of suitable apex vertices
$z$ at which there is a hyperwedge with base $e$, which would complete a
simplex together with $e$. As noted in
\Cref{sec:tech-abundant}, in order to control the variance of the resulting
estimator, we want the vertices of a detected simplex $e \cup \{z\}$ to
respect a certain ordering ``by degree.'' In the more straightforward
algorithm (given in \Cref{sec:subopt}), we use a total ordering of $V$ based
on actual vertex degrees, with ties broken by vertex IDs.  Here, however, this
ordering is more subtle, as hinted in \Cref{sec:techniques}. 

We let the sampled edge $e$ determine the ordering, for which we use {\em
relative} degrees, relative to certain subsets of $e$. Moreover, this
ordering---which we call the $e$-relative ordering---is only partial, as we
never need to compare two vertices that both lie outside $e$. Let $c_i(e)$ be
the $i$th vertex of $e$ according to this ordering that we shall soon define.
In our algorithm, we shall look for a suitable apex $z$ only in the
neighborhood $\Nhd(\{c_1(e),\ldots,c_{k-1}(e)\})$. Furthermore, we shall
require that $z$ have a larger degree than $c_1(e)$, a larger
$\{c_1(e)\}$-relative degree than $c_2(e)$, a larger
$\{c_1(e),c_2(e)\}$-relative degree than $c_3(e)$, and so on.  As we shall
see, the analysis of the resulting algorithm uses the recursive structure of
$k$-graphs and $k$-simplices through the notion of neighborhood hypergraphs
(\Cref{def:nhd-shadow}).

These next two, somewhat elaborate, definitions formalize the above ideas.

\begin{definition}[$\bm{e}$-relative ordering] \label{def:rel-order}
  Fix a hyperedge $e \in E$. The $e$-relative ordering, $\prec_e$, is a
  partial order on $V$ defined as follows. Set $S_0(e) = \varnothing$ and
  define the following, iteratively, for $i$ running from $1$ to $k$.
  \begin{quote}
    Let $c_i(e)$ be the vertex in $e \setm S_{i-1}(e)$ that
    minimizes $\reldeg{S_{i-1}(e)}{c_i(e)}$, with ties resolved in favor
    of the vertex with smallest ID. Set $S_i(e) := 
    S_{i-1}(e) \cup \{c_i(e)\} = \{c_1(e), \ldots, c_i(e)\}$.
  \end{quote}
  Then declare $c_1(e) \prec_e \cdots \prec_e c_k(e)$. Further, for each $z
  \in V \setm e$, declare $c_k(e) \prec_e z$ if the following two things hold:
  (a)~for $1 \le i \le k-1$, we have $\reldeg{S_{i-1}(e)}{c_i(e)} \le
  \reldeg{S_{i-1}(e)}{z}$, and (b)~we also have $\reldeg{S_{k-2}(e)}{c_k(e)}
  \le \reldeg{S_{k-2}(e)}{z}$. Ties in degree comparisons are resolved by
  declaring the vertex with smaller ID to be smaller.%
  \footnote{This is how ties are resolved from now on, even when not
  mentioned.}
\end{definition}

\begin{definition}[Simplex labeling] \label{def:labeling}
  Suppose that $H$ has a simplex at $X$, where $X \subseteq V$.  Number the
  vertices in $X$ as $u_1, \ldots, u_{k+1}$ iteratively, as follows.  For $i$
  from $1$ to $k-1$, let $u_i$ be the vertex in $X$ that minimizes
  $\reldeg{\{u_1, \ldots, u_{i-1}\}}{u_i}$. Let $u_k$ be the vertex among
  $\{u_k, u_{k+1}\}$ that minimizes $\reldeg{\{u_1, \ldots, u_{k-2}\}}{u_k}$.
  This automatically determines $u_{k+1}$.
  Based on this, define the {\em label} of the simplex at $X$ to be
  $(e(X),z(X))$, where $e(X) := \{u_1,\ldots,u_k\}$ and $z(X) := u_{k+1}$.
\end{definition}

Notice that if a simplex is labeled $(e,z)$, then the $e$-relative ordering
satisfies $\max_{\prec_e}(e) \prec_e z$. Conversely, if there is a simplex at
$X$ and an edge $e$ of this simplex such that $X = e \cup \{z\}$ and
$\max_{\prec_e}(e) \prec_e z$, then the label of this simplex is $(e,z)$.

For each $e \in E$, let $\nsimp_e$ denote the number of simplices that have a
label of the form $(e,\cdot)$. Thanks to the uniqueness of labels, we
immediately have
\begin{align}
  \sum_{e \in E} \nsimp_e = T_k(H) \,. \label{eq:sum-by-base}
\end{align}

Our algorithm is described in \Cref{alg:abundant}. Our analysis will make use
the following two combinatorial lemmas at key moments.  We defer the proofs of
these lemmas to \Cref{sec:abundant-proofs}.

\begin{lemma} \label{lem:simplices-on-base}
  \gdef\lemmaSimplicesOnBase{%
  For every edge $e \in E$, we have $\nsimp_e \leq km^{1/k}$.%
  }
  \lemmaSimplicesOnBase
\end{lemma}

\begin{lemma} \label{lem:sum-codeg}
  \gdef\lemmaSumCodeg{%
  Let $k\geq 3$. If $H$ is a $k$-uniform hypergraph, then $\sum_{e \in
  E}\codeg(S_{k-1}(e)) = O\left(m^{1+1/k}\right)$.%
  }
  \lemmaSumCodeg
\end{lemma}

\begin{algorithm*}[!htbp]
\begin{algorithmic}[1]
\Procedure{$k$-Simplex-Count-Abundant}{$\sigma :$ stream of edges of $k$-graph $H = (V,E)$}

\State \textbf{pass 1:} \Comment{$O(1)$ space}
\Indent
	\State pick edge $e=\{u_1,u_2,\ldots,u_k\} \in E$ u.a.r. using reservoir sampling
\EndIndent
\State \textbf{pass 2:} \Comment{$O(2^k) = O(1)$ space}
\Indent
	\State compute $\codeg(S)$ for all non-trivial $S \subset e$
\EndIndent
\State \textbf{pass 3:} \Comment{$O(R)$ space}
\Indent
	\State re-arrange (if needed) $u_1,\ldots,u_k$ so that 
        $u_i = c_i(e)$ for $i \in [k]$, using co-degrees from pass 2
	\State $R \gets \left\lceil \codeg(S_{k-1}(e))\cdot m^{-1/k}\right\rceil$
	        \Comment{note that $\codeg(S_{k-1}(e)) = |\Nhd(S_{k-1}(e))|$}
	\For{$j \gets 1$ to $R$, in parallel,}
		\State $Z_j \gets 0$
		\State pick vertex $x_j \in \Nhd(S_{k-1}(e))$ u.a.r. using reservoir sampling
		from relevant substream of $\sigma$
	\EndFor
\EndIndent
\State \textbf{pass 4:} \Comment{$O(kR) = O(R)$ space}
\Indent
	\State compute the relative degrees $\reldeg{S_{i-1}(e)}{x_j}$ for all $i\in [k-1]$ and $j\in [R]$
	\For {$j \gets 1$ to $R$, in parallel,}
    \If {$e \cup \{x_j\}$ determines a simplex with label $(e,x_j)$} 
    	\State $Z_j \gets \codeg(S_{k-1}(e))$
    \EndIf
	\EndFor
\EndIndent
\State \Return $Y \gets (m/R) \sum_{j=1}^R Z_j$ \Comment{Average out the trials and scale}
\EndProcedure
\end{algorithmic}
\caption{Counting $k$-simplices in $k$-uniform hypergraph streams: the ``abundant'' case}
\label{alg:abundant}
\end{algorithm*}

\subsection{Analysis of the Algorithm} \label{sec:abundant-anal}

We begin by proving that the estimator $Y$ computed by \Cref{alg:abundant} is
unbiased. Let $\cE_e$ denote the event that the edge sampled in pass~1 is $e$.
For each $j \in [R]$, the vertex $x_j$ is picked uniformly at random from
$\Nhd(S_{k-1}(e))$ and exactly $\nsimp_e$ choices cause $Z_j$ to be set to the
nonzero value $\codeg(S_{k-1}(e))$. Therefore,
\begin{align}
  \EE[Z_j\mid \cE_e] 
  = \frac{\nsimp_e}{|\Nhd(S_{k-1}(e))|} \cdot \codeg(S_{k-1}(e)) 
  = \nsimp_e
\end{align}
because $\codeg(S_{k-1}(e)) = |N(S_{k-1}(e))|$, by \Cref{lem:deg-nhd}. By 
the law of total expectation and \cref{eq:sum-by-base},
\begin{align}
  \EE[Y] 
  = \frac{m}{R}\sum_{j=1}^R\sum_{e \in E}
    \frac{1}{m}\EE[Z_j \mid \cE_e] 
  = \frac{1}{R}\sum_{j=1}^R \sum_{e \in E} \nsimp_e 
  = T_k(H) \,.
\end{align}

Next, we bound the variance of the estimator. Proceeding
along similar lines as above, we have
\begin{align}
  \EE[Z_j^2\mid \cE_e] = \codeg(S_{k-1}(e))\cdot \nsimp_e 
  \quad\text{and}\quad 
  \EE[Z_{j_1}Z_{j_2} \mid \cE_e] = \nsimp_e^2 
\end{align}
for all $j_1 \ne j_2$, because $Z_{j_1}$ and $Z_{j_2}$ are independent
conditioned on $\cE_e$. So,
\begin{align*}
  \EE[Y^2\mid \cE_e] 
  &= \EE\left[\left( \frac{m}{R}\sum_{j=1}^R Z_j \right)^2 
    ~\bigg|~ \cE_e\right] \\ 
  &= \frac{m^2}{R^2}\sum_{j=1}^R\EE[Z_j^2 \mid \cE_e] + 
    \frac{m^2}{R^2}\sum_{j_1\neq j_2} \EE[Z_{j_1}Z_{j_2} \mid \cE_e] \\ 
  &\leq \frac{m^2 \cdot \codeg(S_{k-1}(e)) \cdot \nsimp_e}{R} +
    \frac{m^2 R(R-1) \nsimp_e^2}{R^2} \\
  &\leq m^{2+1/k}\nsimp_e + m^2\nsimp_e^2 \,,
\end{align*}
since $\codeg(S_{k-1}(e)) \leq m^{1/k} R$ by the definition of $R$.
Therefore,
\begin{align}
  \Var[Y] \leq \EE[Y^2]
  &= \sum_{e \in E} \frac{1}{m} \EE[Y^2\mid \cE_e] \notag\\
  &\le m^{1+1/k} \sum_{e \in E} \nsimp_e+ m\sum_{e \in E} \nsimp_e^2 \notag\\
  &\le m^{1+1/k}T_k(H) + m\sum_{e \in E} \nsimp_e^2 \,.
  \label{eq:abundant-var-bound}
\end{align}
To bound the final expression in \cref{eq:abundant-var-bound}, we invoke
\Cref{lem:simplices-on-base} and \cref{eq:sum-by-base} to derive
\begin{align*}
  \Var[Y] 
  &\le m^{1+1/k}T_k(H) + m \cdot (\max_{e \in E} \nsimp_e) \cdot
    \sum_{e \in E} \nsimp_e \\
  &= O\left(m^{1+1/k} T_k(H)\right) \,.
\end{align*}

Having bounded the variance, we may apply the median-of-means technique
(\Cref{lem:median-of-means}) to the basic estimator of \Cref{alg:abundant} to
obtain a final algorithm that provides an $(\eps,\delta)$-estimate. As noted
in the comments within \Cref{alg:abundant}, the basic estimator uses $O(R)$
space, where $R$ is a random variable. Therefore, the final algorithm, which
still uses four passes, has expected space complexity $\tO(\EE[R] \cdot
\Var[Y]/\EE[Y]^2) = \EE[R] \cdot \tO(m^{1+1/k}/T_k(H)) = \EE[R] \cdot
\tO(m^{1+1/k}/T)$, thanks to the promise that $T_k(H) \ge T$.

Finally, we need to bound $R$ in expectation. We find that
\begin{align*}
  \EE[R] = \frac{1}{m} \sum_{e \in E} \EE[R \mid \cE_e]
  &= \frac{1}{m} \sum_{e \in E} 
    \left\lceil \codeg(S_{k-1}(e))\cdot m^{-1/k}\right\rceil \\
  &\le 1 + m^{-1-1/k} \sum_{e \in E} \codeg(S_{k-1}(e)) \,.
\end{align*}
Invoking \Cref{lem:sum-codeg}, we conclude that
\begin{align}
  \EE[R] = O(1) \,, \label{eq:r-bound}
\end{align}
which proves our first main algorithmic result.

\begin{theorem} \label{thm:ub-abundant}
  There is a $4$-pass algorithm that reads a stream of $m$ hyperedges
  specifying a $k$-uniform hypergraph $H$ and produces an
  $(\eps,\delta)$-estimate of $T_k(H)$; under the promise that $T_k(H) \ge T$,
  the algorithm runs in expected space $O(\eps^{-2}\log\delta^{-1}\log n\cdot
  m^{1+1/k}/T)$. \qed
\end{theorem}

We conclude with two remarks about this algorithm. First, the space bound can
easily be made worst-case, by applying the following modification to the
median-of-means boosting process. Notice that in each instance of the basic
estimator, at the end of pass~2, we know the value of $R$ and thus the actual
space that this instance would use.  For each batch of estimators for which we
compute a mean, we determine whether this batch would use more than $A$ times
its expected space bound, for some large constant $A$. If so, we abort this
batch and don't use it in the median computation. By Markov's inequality, the
probability that we abort a batch is at most $1/A$, which does not affect the
rest of the standard analysis of the quality of the final median-of-means
estimator.

Second, while we have not dwelt on {\em time} complexity in our exposition, it
is not hard to convince oneself that the algorithm processes each incoming
edge in $\tO(R)$ time, which is $\tO(1)$ in expectation.

\subsection{Proofs of Combinatorial Lemmas} \label{sec:abundant-proofs}

We now prove the two combinatorial lemmas that were used in the above
analysis.

\begin{reminder}[\Cref{lem:simplices-on-base}]
  \lemmaSimplicesOnBase
\end{reminder}
\begin{proof}
  We use induction on $k$. 
  For $k=2$, let $e = \{c_1(e), c_2(e)\}$. Define
  \[
    \Nhd^+(e) = \{z \in V \setm e:\, z \text{ is a neighbor of } c_1(e)
    \text{ and } \deg(z) \ge \deg(c_2(e))\} \,.
  \]
  Then, thanks to \Cref{def:labeling}, $\nsimp_e \le |\Nhd^+(e)|$. Now,
  starting with the handshake lemma, we obtain
  \[
    2m \ge \sum_{z \in \Nhd^+(e)} \deg(z)
    \ge |\Nhd^+(e)| \cdot \deg(c_2(e)) 
    \ge |\Nhd^+(e)| \cdot \deg(c_1(e)) \,,
  \]
  where the final step holds by the definition of $e$-relative ordering
  (\Cref{def:rel-order}). Since $\Nhd^+(e) \subseteq \Nhd(c_1(e))$, we obtain
  $2m \ge |\Nhd^+(e)|^2$. Finally, $\nsimp_e \le |\Nhd^+(e)| \le \sqrt{2m} \le
  2\sqrt{m}$.

  Now let $k \geq 3$. We consider two cases.

  Suppose that $\deg(c_1(e)) \geq m^{(k-1)/k}$. Appealing again to
  \Cref{def:labeling}, in any simplex labeled $(e,z)$, we must have $\deg(z)
  \geq \deg(c_1(e)) \ge m^{(k-1)/k}$. But there are at most $km^{1/k}$
  vertices with degree this high, due to the generalized handshake lemma
  (\Cref{lem:handshake}). It follows that $\nsimp_e \leq km^{1/k}$.

  Suppose, instead, that $\deg(c_1(e)) < m^{(k-1)/k}$. If $X = e \cup \{z\}$
  determines a $k$-simplex labeled with $(e,z)$, then in the neighborhood
  hypergraph $H' := \nhyp{H}{c_1(e)}$, the set $X' = X\setm \{c_1(e)\}$ must
  determine a $(k-1)$-simplex.  Further, $X'$ must be labeled with
  $(e\setm\{c_1(e)\},z)$ in $H'$ because, for all $S$, $S$-relative degrees in
  $\nhyp{H}{c_1(e)}$ are $(\{c_1(e)\} \cup S)$-relative degrees in $H$. It
  follows that $\nsimp_e \leq \nsimp_{e\setm \{c_1(e)\}}^{H'}$.

  Using the induction hypothesis on the $(k-1)$-graph $H'$, we
  have $\nsimp_{e\setm\{c_1(e)\}}^{H'} \leq (k-1)|E(H')|^{1/(k-1)}$. But
  $|E(H')| = \deg(c_1(e))$, yielding
  \[
    \nsimp_e \leq (k-1) \cdot \deg(c_1(e))^{1/(k-1)} < km^{1/k} \,. \qedhere
  \]
\end{proof}

\begin{reminder}[\Cref{lem:sum-codeg}]
  \lemmaSumCodeg
\end{reminder}
\begin{proof}
  We again use induction on $k$. 
  For $k=2$, each set $S_{k-1}(e)$ is a singleton, consisting of the
  smaller-degree endpoint of $e$. The result now follows by combining
  Lemmas~1(a) and~2 in \cite{ChibaN85}: specifically, we obtain
  $\sum_{\{u,v\} \in E} \min\{\deg(u), \deg(v)\} \le 2m^{3/2}$.

  For $k \geq 3$, consider partitioning the vertex set into ``light'' and
  ``heavy'' vertices as follows:
  \begin{align*}
    \light &:= \{u \in V:\, \deg(u) < m^{(k-1)/k}\} \,; \\
    \heavy &:= \{u \in V:\, \deg(u) \geq m^{(k-1)/k}\} \,.
  \end{align*}

  Let $M_u := \{e \in E\mid u = c_1(e)\}$ be the set of edges in which
  vertex $u$ has the minimum degree. Then $E = \bigcup_{u \in V} M_u$, whence
  \begin{align}
    \sum_{e\in E}\codeg(S_{k-1}(e)) 
    &= \sum_{u \in \light} \sum \limits_{e \in M_u} \codeg(S_{k-1}(e)) +
      \sum_{u \in \heavy} \sum_{e \in M_u} \codeg(S_{k-1}(e)) \,.
      \label{eq:sum-codeg}
  \end{align}

  We first take care of the first term above: the sum over $u \in \light$.
  Each edges $e \in M_u$ corresponds to an edge $e' = e \setm \{u\}$
  in the $(k-1)$-uniform hypergraph $\nhyp{H}{u}$. Edge $e'$ has $k-1$
  vertices and we have that $S_{k-2}(e') \cup \{u\} = S_{k-1}(e)$. Therefore,
  $\codeg_H(S_{k-1}(e)) = \codeg_{\nhyp{H}{u}}(S_{k-2}(e'))$ and we have
  \begin{align*}
    \sum_{u \in \light} \sum_{e \in M_u} \codeg(S_{k-1}(e)) 
    \leq \sum_{u \in \light} \sum_{e' \in E(\nhyp{H}{u})} 
      \codeg_{\nhyp{H}{u}}(S_{k-2}(e')) \,.
  \end{align*}
  Applying the induction hypothesis, we obtain
  \begin{align*}
    \sum_{e' \in E(\nhyp{H}{u})} \codeg_{\nhyp{H}{u}}(S_{k-2}(e')) 
    &= O\left(|E(\nhyp{H}{u})|^{k/(k-1)}\right) \\
    &= O(\deg(u)^{k/(k-1)}) \,,
  \end{align*}
  implying
  \begin{align}
    \sum_{u \in \light} \sum_{e' \in E(\nhyp{H}{u})} \codeg_{\nhyp{H}{u}}(S_{k-2}(e')) 
    &= O(1) \cdot \sum_{u \in \light}\deg(u)^{k/(k-1)} \notag \\
    &= O(1) \cdot \max_{u \in \light} \left( \deg(u)^{1/(k-1)} \right) \cdot
      \sum_{u \in \light}\deg(u) \notag \\
    &= O(1) \cdot m^{1/k} \cdot km 
    = \textstyle O\left(m^{1+1/k}\right) \,.
      \label{eq:sum-codeg-light}
  \end{align}
 
  Next, we want to bound the second term in \cref{eq:sum-codeg}: the sum over
  $u \in \heavy$.

  Let $M$ be the multiset whose elements are the $(k-1)$-sized subsets of $V$
  whose codegrees appear in the sum over $u \in \heavy$. Let $M'$ be the
  support of $M$. If $P = \{u_1,\ldots,u_{k-1}\}$ is one subset in $M$, let
  $f_P$ be its multiplicity: the number of edges $e$ which are such that
  $c_1(e) \in \heavy$ and $S_{k-1}(e) = P$. We write
  \begin{align*}
    M = \bigcup_{P\in M'}P^{(f_P)} \,,
  \end{align*}
  a multiset union. 
  We claim that $f_P = O\left(m^{1/k}\right)$ for all $P \in M$. For some
  hyperedge we have that $P = S_{k-1}(e)$. Let $z = e \setm P$ be the
  vertex of $e$ not in $P$. Then, by definition, $\deg(z) \geq
  \deg(c_1(e)) \geq m^{(k-1)/k}$. So we have
  \[
    f_P \leq \left|\left\{z \in V:\, \deg(z) 
    \geq m^{(k-1)/k)}\right\}\right| 
    = O(m^{1/k})
  \]
  by the handshake lemma. Using this bound on $f_P$ and the handshake lemma
  again, we obtain
  \begin{align}
    \sum_{u\in \heavy}\sum_{e \in {M}_u} \codeg(S_{k-1}(e)) 
    &= \sum_{P \in M} \codeg(P) \notag\\
    &= \sum_{P \in M'} \codeg(P)\cdot f_P \notag\\ 
    &= O(m^{1/k}) \cdot \sum_{P \in M'} \codeg(P) \notag\\
    &\leq O(m^{1/k})\cdot km = O(m^{1+1/k}) \,. 
      \label{eq:sum-codeg-heavy}
  \end{align}

  The proof is completed by combining
  \cref{eq:sum-codeg,eq:sum-codeg-light,eq:sum-codeg-heavy}.
\end{proof}

\section{Algorithms Based on Oblivious Sampling} \label{sec:ub-meager}

We now turn to our second family of algorithms, which use oblivious sampling,
as outlined in \Cref{sec:tech-meager}. To recap, whereas the algorithms in the
previous section used an initially-sampled edge to guide their subsequent
sampling, an oblivious sampling strategy decides whether or not to store an
edge based only coin tosses specific to that edge and perhaps its constituent
vertices (and not on any edges in store).  The main result of this section is
an $\tO(m/T^{2/(k+1)})$-space algorithm, given as \Cref{alg:meager-best},
which is the method of choice for the ``meager'' case, when $T$ is not too
large (specifically, $T < m^{(k+1)/(k^2-k)}$). Along the way, we give two
other algorithms that have worse space guarantees but help build up the ideas
towards \Cref{alg:meager-best}. In the sequel, we give a $1$-pass algorithm
whose space guarantee depends on additional structural parameters of the input
hypergraph $H$. 

\subsection{A Unifying Framework for Oblivious Sampling Strategies}
\label{sec:obliv-framework}

\Cref{sec:tech-meager} outlined a framework for analyzing algorithms based on
oblivious sampling.  This framework streamlines the presentations of the
algorithms below and is also valuable because it unifies a number of
approaches seen in previous work on triangle counting \cite{McGregorVV16,
JayaramK21,KallaugherP17,BravermanOV13}.  We now spell out the details.

Suppose that we are trying to estimate the number, $T(H)$, of copies of a
target substructure (motif) $\Xi$ in a streamed hypergraph $H$. Suppose we do
this by collecting certain edges into a random sample $Q$, using some random
process parameterized by a quantity $p$. The sample $Q$ determines which of
the $T(H)$ copies of $\Xi$ get ``detected'' by the rest of the logic of the
algorithm. Let $\Xi_1,\ldots,\Xi_{T(H)}$ be the copies of $\Xi$ in $H$ (under
some arbitrary enumeration scheme) and let $\Lambda_i$ be the indicator random
variable for the event that $\Xi_i$ is detected. 

In our framework, we will want to identify parameters $\alpha, \beta$, and
$\gamma$ such that the following conditions hold for all $i, j \in [T(H)]$
with $i \ne j$ and all $e \in E$: 
\begin{align}
  \Pr[\Lambda_i = 1] &= p^\alpha \,, \label{eq:alpha-def} \\
  \Pr[\Lambda_i = \Lambda_j = 1] &\le p^\beta \,, \label{eq:beta-def} \\
  \Pr[e \in Q] &= p^\gamma \,, \label{eq:gamma-def}
\end{align}
In instantiations of the framework, we will of course have $\beta > \alpha$
and further, we will have $\beta < 2\alpha$, which is natural because for
certain pairs $i,j$, the variables $\Lambda_i$ and $\Lambda_j$ will be
positively correlated. We shall further require that our detection process
satisfy the {\em separation property}, where $\Xi \cap \Xi'$ denotes the set
of {\em edges} common to $\Xi$ and $\Xi'$ and the notation ``$X \perp Y$''
means that random variables $X$ and $Y$ are independent.

\begin{align}
  \textsc{Separation Property:~} \text{If } \Xi_i \cap \Xi_j = \varnothing, 
  \text{ then } \Lambda_i \perp \Lambda_j \,. \label{eq:sep}
\end{align}
The logic of the algorithm will count the number of detected copies of $\Xi$
in an accumulator $A := \sum_{i=1}^{T(H)} \Lambda_i$. Defining the estimator
$\tT := A/p^\alpha$, we see that
\begin{align}
  \EE\left[ \frac{A}{p^\alpha} \right] 
  = \sum_{i=1}^{T(H)} \frac{1}{p^\alpha} \EE[\Lambda_i] 
  = T(H) \,,
\end{align}
which makes it unbiased. To establish concentration, we estimate the variance
of $A$ as follows.
\begin{align}
  \Var[A] 
  &= \Var\left[\sum_{i=1}^{T(H)} \Lambda_i\right] \notag\\
  &= \sum_{i=1}^{T(H)} \Var[\Lambda_i] + \sum_{i \ne j}\Cov[\Lambda_i,\Lambda_j] \notag\\
  &\le \sum_{i=1}^{T(H)} \EE[\Lambda_i^2] + \sum_{i\ne j}\Cov[\Lambda_i,\Lambda_j] \notag\\
  &\le T(H)p^\alpha + \sum_{\Lambda_i \nperp \Lambda_j}\EE[\Lambda_i\Lambda_j] \notag\\
  &\le T(H) p^\alpha + \sum_{\Lambda_i \nperp \Lambda_j} p^\beta \label{eq:detect-two} \\
  &= T(H) p^\alpha + p^\beta \cdot \underbrace{\sum_{e \in E}
    \sum_{\Xi_i \cap \Xi_j = \{e\}} 1}_{\scor} \label{eq:var-bound} \,,
\end{align}
where \eqref{eq:detect-two} follows from \eqref{eq:beta-def} and
\eqref{eq:var-bound} holds because of the separation property and the
reasonable structural assumption (true of simplices) that two distinct copies 
of the motif $\Xi$ can intersect in at most one edge. 

The term $\scor$ captures the amount of correlation in the sampling process.
Let us now specialize the discussion to the problem at hand, where the motif
$\Xi$ is a $k$-simplex and so, $T(H) = T_k(H)$. Let $\Nsimp_e$ be the number
of simplices that share hyperedge $e$ (note that this is subtly different from
$\nsimp_e$ in \Cref{sec:ub-abundant}). Recalling the definition of $\Delta_E$
in \cref{eq:delta_e}, we see that $\max_{e\in E} \Nsimp_e = \Delta_E$.  We can
therefore upper bound $\scor$ as follows:
\begin{align}
  \scor
    \le \sum_{e \in E} \binom{\Nsimp_e}{2} 
    \le \sum_{e \in E} \frac{\Nsimp_e^2}{2}
    \le \frac{\Delta_E}{2} \sum_{e \in E} \Nsimp_e 
    = \frac{(k+1)\Delta_E T_k(H)}{2} \,.
    \label{eq:s_cor_bound}
\end{align}
Plugging this bound into \eqref{eq:var-bound} and applying Chebyshev's
inequality gives the following ``error probability'' bound, where $B > 0$ will
be chosen later.
\begin{align}
  \Perr^{(B)}\big[\, \tT \,\big] 
  &:= \Pr\left[ |\tT-T(H)| \geq \eps B \right] \notag\\
  &\le \frac{\Var[A]}{p^{2\alpha} \eps^2 B^2} \notag\\
  &\le \frac{T_k(H)}{\eps^2 p^\alpha B^2} + 
    \frac{(k+1)\Delta_E T_k(H)}{2\eps^2 p^{2\alpha-\beta} B^2} \,.
    \label{eq:error_common_edge}
\end{align}
It may be helpful to imagine $B \approx T_k(H)$ and compare the above bound to
its informal version in \cref{eq:estim-deviates}.

\paragraph{Heavy/Light Edge Partitioning via an Oracle.}
Using the above estimator $\tT$ directly on the input hypergraph $H$ does not
give a good bound in general, because the quantity $\Delta_E$ appearing in
\cref{eq:error_common_edge} could be as large as $T_k(H)$, making the upper
bound on $\Perr$ useless. One remedy is to apply the technique so far to a
subgraph of $H$ where $\Delta_E$ is {\em guaranteed} to be a fractional power
of $T_k(H)$ and deal separately with simplices that don't get counted in such
a subgraph. Such a subgraph can be obtained by taking only the ``light'' edges
in $H$, where an edge $e$ is considered to be light if $\Nsimp_e =
O(T_k(H)^{\theta})$, for some parameter $\gamma < 1$ that we will soon choose.
We remark that this kind of heavy/light partitioning is a standard technique
in this area: it occurs in the triangle counting algorithms of 
\cite{CormodeJ14, McGregorVV16} and in fact in the much earlier triangle
listing algorithm of \cite{ChibaN85}.

More precisely, we create an {\em oracle} to label each edge as either \heavy
or \light. We then have
\begin{align}
  T_k(H) = \Tklight(H) + \Tkheavy(H) \,, \label{eq:heavy_light}
\end{align}
where $\Tklight(H)$ is the number of simplices containing only light edges,
and $\Tkheavy(H)$ is the number of simplices containing at least one heavy
edge. The key insight in the partitioning technique is that we can estimate
the two terms in \cref{eq:heavy_light} separately, using the estimator $\tT$
described above for $\Tklight(H)$ and a different estimator for $\Tkheavy(H)$
that we shall soon describe. In what follows, we allow an additive error of
$\pm \eps T_k(H)$ in the estimate of each term, leading to a multiplicative
error of $1\pm 2\eps$ in the estimation of $T_k(H)$.

The oracle is implemented by running the following one-pass randomized
streaming algorithm implementing a function $\oracle \colon E \to
\{\heavy,\light\}$ to label the edges of $H$.  Form a random subset $Z
\subseteq V$ by picking each vertex $v \in V$, independently, with probability
\begin{align}
  q := \xi\eps^{-2}\ln n \cdot T^{-\theta} \,,
  \label{eq:q-def}
\end{align}
where $T$ is the promised lower bound on $T_k(H)$ given to the algorithm, and
$\xi \geq 12k(k+1)$ is a constant.  Then, in one pass over the input stream,
collect all hyperedges $e$ containing at least one vertex from $Z$. Let $S$
denote the random sample of edges thus collected. Note that this is distinct
from (and independent of) the sample $Q$ collected by the oblivious sampling
scheme discussed above.  Now, for each hyperedge $e = \{u_1,\ldots,u_k\} \in
E$, let
\begin{align}
    \widetilde{\Nsimp}_e := |\{z \in Z\,:\,e\cup\{z\} \in \cT_k(H)\}|
\end{align}
be the number of simplices that $e$ completes with respect to $Z$. Then we define: 
\begin{align}
    \oracle(e) = \begin{cases}
        \heavy \,, & \text{if } \widetilde{\Nsimp}_e < qT^{\theta} \,, \\
        \light \,, & \text{otherwise.}
    \end{cases}
    \label{eq:oracle}
\end{align}

The next lemma says that the oracle's predictions correspond, with high
probability, to our intuitive definition of heavy/light edges, up to a factor
of $2$. 

\begin{lemma} \label{lem:oracle}
  With high probability, the oracle given by \eqref{eq:oracle} is ``correct,''
  meaning that it satisfies
  \begin{align}
    \forall\, e \in E:~
    &\oracle(e) = \light ~\Rightarrow~ \Nsimp_e < 2T^\theta ~\wedge~ \notag\\
    &\oracle(e) = \heavy ~\Rightarrow~ \Nsimp_e > \tfrac12 T^\theta \,.
    \label{eq:oracle_correct}
  \end{align}
\end{lemma}
\begin{proof}
  Pick a particular edge $e$. Observe that $\widetilde{\Nsimp}_e$ follows a
  binomial distribution: $\widetilde{\Nsimp}_e\sim\Bin(\Nsimp_e,q)$. In the
  case when $\Nsimp_e \geq 2T^{\theta}$, one form of the Chernoff bound gives
  \begin{align*}
      \Pr\big[\oracle(e) = \light\big] 
      &= \Pr\left[\widetilde{\Nsimp}_e < (1-\tfrac12) q\cdot 2T^{\theta}\right] \\
      &\le \exp\left(-2qT^{\theta}\cdot (\tfrac12)^2/2\right) 
      \le n^{-2k} \,,
  \end{align*}
  where, in the final step, we used $\xi \ge 8k$ and $\eps \le 1$.
  Analogously, in the case when $\Nsimp_e \le \frac12 T^\theta$, we can use
  another form of the Chernoff bound to get
  \begin{align*}
      \Pr\big[\oracle(e) = \heavy\big] 
      &= \Pr\left[\widetilde{\Nsimp}_e \ge qT^{\theta}\right] \\
      &\le \exp\left(-qT^\theta \ln(2(1-q))\right)
      \le n^{-2k} \,.
  \end{align*}
  Taking a union bound over the at most $n^k$ edges $e$ in the hypergraph, the
  probability that the oracle is incorrect, i.e., that
  condition~\eqref{eq:oracle_correct} does not hold, is at most $n^{-O(1)}$.
\end{proof}

Notice that the space required to implement the oracle is that required to
store the sample $S$. Since 
\begin{align}
  \EE[|S|] = O\left(\sum_{v \in V} q\deg(v)\right) = O(qm) \,,
  \label{eq:s-size}
\end{align}
the expected space requirement is $O(qm \log n) = O(\eps^{-2}\log^2 n \cdot
T^{-\theta} m) = \tO(m/T^\theta)$, by \cref{eq:q-def}.

\paragraph{Counting the Light and Heavy Simplices.}
We can now describe our overall algorithm. We use two streaming passes. In the
first pass, we collect the samples $S$ and $Q$. At the end of the pass, we
have the information necessary to prepare the heavy/light labeling oracle and
using it, we remove all heavy edges from $Q$. In the second pass, for each
edge $e$ encountered, we use the oracle to classify it. If $e$ is light, we
feed it into the detection logic for the oblivious sampling scheme (which uses
$Q$). If $e$ is heavy, we use $S$ to estimate $\Nsimp_e$, the number of
simplices containing $e$; this requires a little care, as we explain below.

In analyzing the algorithm, we start by assuming that the oracle is correct
(which happens w.h.p., by \Cref{lem:oracle}).  Let $\tT^L$ be the estimator
for $\Tklight(H)$ produced by the oblivious sampling scheme. We analyze its
error probability using \cref{eq:error_common_edge}, noting that $\Delta_E \le
2T^\theta$ in the light-edges subgraph and choosing $B = T_k(H)$. This leads
to 
\begin{align}
  \Perr[\tT^L]
  &:= \Pr\left[ |\tT^L - \Tklight(H)| \ge \eps T_k(H) \right] \notag\\
  &\le \frac{\Tklight(H)}{\eps^2 p^\alpha T_k(H)^2} + 
    \frac{(k+1)\cdot 2T^\theta\cdot \Tklight(H)}{2 \eps^2 p^{2\alpha-\beta} T_k(H)^2} \notag \\
  &\le \frac{1}{\eps^2 p^\alpha T} + 
    \frac{k+1}{\eps^2 p^{2\alpha-\beta} T^{1-\theta}} \,,
    \label{eq:light_error}
\end{align}
where we used the inequalities $\Tklight(H) \le T_k(H)$ and $T_k(H) \ge T$.
This bound on the error probability will inform our choices of parameters $p,
\alpha, \beta,$ and $\theta$, below.

Finally, we estimate $\Tkheavy(H)$. Define a simplex in $\cT_k(H)$ to be {\em
$i$-heavy} if it has exactly $i$ heavy edges. For each edge $e$ and $1 \le i
\le k+1$, let $\Nsimp_e^i$ be the number of $i$-heavy simplices containing
$e$.  Notice that
\begin{align}
  \sum_{e \text{ heavy}} \Nsimp_e^i 
  &= i \cdot |\{\Xi \in \cT_k(H):\, \Xi \text{ is $i$-heavy}\}| \,, 
  \text{\rlap{~~whence}} \notag\\
  \Tkheavy(H) 
  &= \sum_{e \text{ heavy}} \sum_{i=1}^{k+1} \frac{\Nsimp_e^i}{i} \,.
  \label{eq:heavy_sum}
\end{align}
During our second streaming pass, we compute unbiased estimators for each
$\Nsimp_e^i$ as follows.  Recall that $S$ is the set of all edges incident to
some vertex in the set $Z$ of sampled vertices. Define
$\widetilde{\Nsimp}_e^i$ to be the number of vertices in $Z$ that combine with
$e$ to form an $i$-heavy simplex; we consult the oracle to test whether a
simplex in question is indeed $i$-heavy. Notice that $\widetilde{\Nsimp}_e^i
\sim \Bin(\Nsimp_e^i,q)$. Therefore, by \cref{eq:heavy_sum},
\begin{align}
  \tT^H := \frac{1}{q}\sum_{e \text{ heavy}}\sum_{i=1}^{k+1}
    \frac{\widetilde{\Nsimp}_e^i}{i} \label{eq:heavy-est}
\end{align}
is an unbiased estimator for $\Tkheavy(H)$. To establish concentration, pick
any particular heavy edge $e$ and consider the inner sum corresponding to this
$e$. The sum can be rewritten as a sum of $\Nsimp_e$ mutually independent
indicator variables scaled by factors in $\{1, 1/2, \ldots, 1/(k+1)\}$.
Since the oracle is correct, the heaviness of $e$ implies that $\Nsimp_e \ge
\frac12 T^\theta$ and then a Chernoff bound gives
\begin{align*}
  \Pr\left[\sum_{i=1}^{k+1} \frac{\widetilde{\Nsimp}_e^i}{i} \notin 
    (1\pm\eps) q\sum_{i=1}^{k+1} \frac{\Nsimp_e^i}{i}\right] 
  &\le 2\exp\left(-\frac13 \eps^2 q \cdot 
    \frac12 T^\theta \cdot \frac{1}{k+1}\right) \\
  &\le n^{-2k} \,,
\end{align*}
where the final step uses $\xi \ge 12k(k+1)$.
Applying a union bound over the at most $n^k$ heavy edges gives
\begin{align}
  \Perr[\tT^H]
  &:= \Pr\left[ |\tT^H - \Tkheavy(H)| \ge \eps T_k(H) \right] \notag \\
  &\le \Pr\left[ \tT^H \notin (1\pm\eps) \Tkheavy(H) \right] \notag \\
  &\le \Pr\left[ \bigvee_{e \text{ heavy}} \left\{
    \sum_{i=1}^{k+1} \frac{\widetilde{\Nsimp}_e^i}{i} \notin 
      (1\pm\eps) q\sum_{i=1}^{k+1} \frac{\Nsimp_e^i}{i} \right\} \right]
  \le n^{-O(1)} \,.
  \label{eq:heavy_error}
\end{align}

\paragraph{Determining the Parameters and Establishing a Space Bound.}
We have now described the overall logic of a generic algorithm in our
framework.  To instantiate it, we must plug in a specific oblivious sampling
scheme, whose logic will determine the parameters $\alpha, \beta$, and
$\gamma$ as given by \cref{eq:alpha-def,eq:beta-def,eq:gamma-def}.  Now we
choose $p$ to ensure that $\Perr[\tT^L]$ is less than a small constant.
Thanks to \cref{eq:light_error}, it suffices to ensure that
\[
  \frac{1}{\eps^2 p^\alpha T} + \frac{k+1}{\eps^2 p^{2\alpha-\beta}
  T^{1-\theta}} = o(1) \,.
\]
Setting $p = (\log n / (\eps^2T))^{1/\alpha}$ reduces the first term to
$1/\log n$. To make the second term small as well, we set the threshold
parameter $\theta$ used by the heavy/light partitioning oracle to $\theta :=
\beta/\alpha - 1$; as we remarked after \cref{eq:gamma-def}, we will have
$\alpha < \beta < 2\alpha$, which implies $0 < \theta < 1$. The second term
now reduces to
\[
  \frac{(k+1)p^{\beta-\alpha}T^\theta}{\log n}
  = \frac{k+1}{(\log n)^{1-\theta} \eps^{2\theta}} = o(1) \,,
\]
for $n$ sufficiently large. We therefore obtain $\Perr[\tT^L] = o(1)$ and,
putting this together with \cref{eq:heavy_error} and \Cref{lem:oracle}, we
conclude that the overall algorithm computes a $(1\pm2\eps)$-approximation to
$T_k(H)$ with probability at least $1 - o(1)$.

With these parameters now set, the space complexity analysis is
straightforward. The space usage is dominating by the storing of the samples
$S$ and $Q$. By \cref{eq:s-size} and \cref{eq:q-def}, we have $\EE[|S|] =
\tO(m/T^\theta)$ and by \cref{eq:gamma-def}, we have
\[
  \EE[|Q|] 
  = p^\gamma m 
  = \left(\frac{\log n}{\eps^2}\right)^{\gamma/\alpha}
    \frac{m}{T^{\gamma/\alpha}} \,.
\]
It is natural that the probability of detecting a particular simplex is at
most the probability of storing one particular edge, so $\gamma \le \alpha$.
Therefore, $\EE[|Q|] \le \eps^{-2}\log n\cdot m/T^{\gamma/\alpha}$.
Altogether, considering that storing a single edge uses $O(\log n)$ bits, the
total space usage is $\tO(m/T^\lambda)$, where
\begin{align}
  \lambda 
  = \min\left\{ \frac{\gamma}{\alpha},\, \theta \right\}  % I do want to show this step
  = \min\left\{ \frac{\gamma}{\alpha},\, \frac{\beta}{\alpha} - 1 \right\} \,.
    \label{eq:lambda-def}
\end{align}

This completes the description of our framework for designing and analyzing
simplex counting algorithms based on oblivious sampling. Next, we instantiate
the framework and obtain concrete algorithms.

\subsection{Simplest Instantiation: Sampling Edges Independently}
\label{sec:ub-meager-simplest}

Instead of going straight to our best algorithm in this family, we shall build
up to it to drive home the point that our framework allows us to quickly
design and analyze an oblivious-sampling-based algorithm. 

For our first algorithm, as outlined in \Cref{sec:tech-meager}, we directly
generalize the $2$-pass, $\tO\left(m/\sqrt{T}\right)$-space algorithm of
\cite{McGregorVV16}. The oblivious sample $Q$ is obtained by choosing each
edge with probability $p$, independently. We then detect a light simplex when
one of its edges appears in the stream and the $k$ edges forming the remaining
hyperwedge are all in $Q$. Notice that each $k$-simplex in $H$ can be detected
in exactly $k+1$ ways, so the expectation of the resulting estimator is
exactly $(k+1)T_k(H)$. To fit this algorithm into our framework, we take
$\Xi$, the motif being counted, to be a hyperwedge arising from a $k$-simplex
in $H$. 

We must now determine $\alpha, \beta$, and $\gamma$ based on
\cref{eq:alpha-def,eq:beta-def,eq:gamma-def} and check that the separation
property \eqref{eq:sep} holds. The probability of detected any particular copy
of $\Xi$ is $p^k$, so $\alpha = k$. Two hyperwedges are correlated iff they
have an edge $e$ in common, so the separation property holds. Further, when
such correlation exists, the probability that they are simultaneously detected
is $p^{2k-1}$, so $\beta = 2k-1$. Since each edge is stored with probability
$p$, we have $\gamma = 1$. Therefore, \cref{eq:lambda-def} gives us $\lambda =
1/k$ and we obtain the following result.

\begin{theorem} \label{thm:ub-meager-worst}
  We can estimate $T_k(H)$ using a $2$-pass algorithm that runs in
  $\tO(m/T^{1/k})$ space. \qed
\end{theorem}

\subsection{Better Instantiation: Sampling by Coloring}
\label{sec:ub-meager-better}

Our second idea, also outlined in \Cref{sec:tech-meager}, is inspired by the
graph algorithm of \cite{PaghT12} that colors vertices independently and
samples monochromatic edges. As noted there, this idea does not translate
directly to the hypergraph setting, because this sampling scheme breaks the
separation property. So we instead uniformly and independently \emph{color}
$(k-1)$-sized subsets of $V$ and select an edge into the sample $Q$ iff all
$k$ of its subsets receive the same color. We call such an edge
\emph{monochromatic}.

\begin{algorithm*}[!htb]
\begin{algorithmic}[1]
	\State $d \gets 2\binom{k+1}{k-1}$;~ $r \gets \binom{n}{k-1}$;~
                $\alpha \gets \binom{k+1}{k-1}-1$;~
		$p \gets (\eps^{-2} T^{-1} \log n)^{1/\alpha}$;~ $N \gets \ceil{1/p}$
	\State let $\cF \subseteq [N]^{[r]}$ be a $d$-wise independent family of hash functions,
                each of the form $h\colon [r]\to [N]$
	\State select a function $\col(~)$ uniformly at random from $\cF$ to color the elements of $\binom{V}{k-1}$
	\State when the stream arrives, collect all edges that are monochromatic into a set $Q$
	\State $\tT \gets$ number of simplices induced by $Q$
	\State \Return $\tT/p^\alpha$
\end{algorithmic}
\caption{Simplex counting in the ``meager'' case using $O(m/T^{2/(k+2)})$ space: estimating $\Tklight(H)$}
\label{alg:meager-better}
\end{algorithm*}

If $N$ colors are used, then the probability that any $(k-1)$-sized subset is
colored with any particular color is $p = 1/N$. Then an edge is sampled with
probability $N(1/N)^k = p^{k-1}$, making $\gamma = k-1$. A particular simplex
in $H$ is detected when all $\binom{k+1}{k-1}$ of the involved subsets are
colored the same, which happens with probability $p^\alpha$ for
\[
  \alpha = \binom{k+1}{k-1}-1 = \frac{(k+2)(k-1)}{2} \,.
\]

Now consider a pair of distinct simplices $\Xi, \Xi'$ in $H$. If $\Xi$ and
$\Xi'$ do not share an edge, then they may have at most $k-1$ vertices in
common. Hence, it is easy to see that the probability that both $\Xi$ and
$\Xi'$ are detected is $p^{2\alpha}$, so there is no correlation; the
separation property holds. If, on the other hand, $\Xi$ and $\Xi'$ have an
edge in common, then they share $k$ common $(k-1)$-sized subsets, so the
probability that both $\Xi$ and $\Xi'$ are detected is $p^\beta$, where
\[
  \beta = 2\binom{k+1}{k-1}-k-1 = k^2-1 \,.
\]
Therefore, \cref{eq:lambda-def} gives us $\lambda = \gamma/\alpha = 2/(k+2)$.

In the above description, we took the coloring function to be uniformly
random. However, as the analysis shows, all we need is that the above
probability calculations hold. For this, it suffices for the coloring function
to be independent on the family of all $(k-1)$-sized subsets within any two
distinct simplices. The number of such subsets is at most $d :=
2\binom{k+1}{k-1}$, so we may use a $d$-wise independent family of hash
functions to provide the coloring. This is shown in the formal description of
the algorithm (see \Cref{alg:meager-better}). For ease of presentation, we
show only the light-simplex counting portion and not the final algorithm,
which would use the heavy/light partitioning technique on top of this.

%%% The algorithm float was given earlier in the subsection 

Since we obtained $\lambda = 2/(k+2)$, we have the following result.

\begin{theorem} \label{thm:ub-meager-better}
  We can estimate $T_k(H)$ using a $2$-pass algorithm that runs in
  $\tO\left(m/T^{2/(k+2)}\right)$ space. \qed
\end{theorem}

\subsection{Our Method of Choice: Oblivious Sampling Using the Shadow Hypergraph}
\label{sec:ub-meager-best}

We are ready to introduce our best algorithm in the oblivious sampling
framework. It achieves a space complexity of $\tO(m/T^{2/(k+1)})$, which is
our best bound of the form $\tO(m/T^\lambda)$.  It is therefore our method of
choice for the ``meager'' case. Comparing with \Cref{thm:ub-abundant}, we see
that this bound wins when $T < m^{(k+1)/(k^2-k)}$.

For this, we again use our unified framework for the main design and analysis.
We use the above idea of coloring appropriately-sized subsets of vertices. The
added twist is that we apply coloring to vertices of the {\em shadow
hypergraph} $H' = (V',E')$, defined in \Cref{def:nhd-shadow}.

To be precise, we uniformly and independently color $(k-2)$-sized subsets of
vertices $V'$ of $H'$ using $N$ colors. Since each vertex in $V'$ is
essentially an ordered pair of vertices in $V$, we need a coloring function
that maps $\binom{V\times V}{k-2}$ to $[N]$. Recall that $H' = (V',E')$ is a
$(k-1)$-uniform hypergraph. During the first pass, when an edge $e =
\big\{u_1,\ldots,u_k\big\}$ arrives in the stream, it implies the arrival of
an edge $e' = E'$ which we store in $Q$ iff all $(k-2)$-sized subsets within
$e'$ receive the same color. We call such an edge $e'$ {\em monochromatic}.

We detect simplices using this sample $Q$ as follows.  As observed after
\Cref{def:hyperwedge}, each $k$-simplex in $H$ corresponds to exactly one
$(k-1)$-simplex in $H'$; we try to detect this ``shadow simplex.'' In greater
detail, let $\Xi$ be a particular simplex in $H$. Let $z$ be its minimum-ID
vertex, let $e_z$ be the edge of $\Xi$ not incident to $z$ and let $W_z$ be
the hyperwedge of $\Xi$ with apex $z$. Note that the edges of $W_z$ correspond
to edges in $H'$ that lie in $E'_z$ (see \Cref{def:nhd-shadow}) and, since $W_z$
is a hyperwedge, these edges form a $(k-1)$-simplex in the $z$-flavored
component of $H'$. In the second pass, when $e_z$ arrives in the stream, we
detect $\Xi$ iff this entire $(k-1)$-simplex has been stored in $Q$, i.e., iff
all its edges are monochromatic. Algorithmically, when we see an edge $e$ in
the second pass, we go over all vertices $z \in V$, detecting a simplex
whenever this $e$ plays the role of $e_z$ for a simplex on the vertices in $e
\cup \{z\}$ in the manner just described.

The remaining details are as before: we apply the heavy/light edge
partitioning technique on top of the above sampling scheme and proceed as in
our framework. We formalize the overall algorithm, with full details, as
\Cref{alg:meager-best}. As in \Cref{sec:ub-meager-better}, the coloring
function need not be fully random and can be chosen from a $d$-wise
independent hash family for an appropriate constant $d$.  The update logic for
the accumulator $A_L$ (\cref{line:count-light}) implements the detection
method just described. The update logic for the accumulator $A_H$
(\cref{line:count-heavy}) implements the estimator in \cref{eq:heavy-est}.
Note that the heavy-simplex estimator does not use the shadow hypergraph.

%%% The algorithm float, placed here, works well in both one- and two-column formats!

\begin{algorithm*}[!th]
\begin{algorithmic}[1]
\Procedure{$k$-Simplex-Count-Meager}{$\sigma :$ stream of edges of $k$-graph $H = (V,E)$}

	\State $d \gets 2\binom{k}{k-2}$;~ $r \gets \binom{n^2}{k-2}$;~
                $\alpha \gets \binom{k}{k-2}-1$;~
		$p \gets (\eps^{-2} T^{-1} \log n)^{1/\alpha}$;~ $N \gets \ceil{1/p}$
        \State $\theta \gets \big(2\binom{k}{2}-k\big) / \big(\binom{k}{2}-1\big)$;~
                $q \gets \xi \eps^{-2}\log n \cdot T^{-\theta}$ 
                \Comment{$\xi$ is a suitable constant; see \cref{eq:q-def}}
	\State $Z \gets$ sample of $V$, picking each vertex independently with probability $q$
	\State let $\cF \subseteq [N]^{[r]}$ be a $d$-wise independent family of hash functions,
                each of the form $h\colon [r]\to [N]$
	\State select a function $\col(~)$ uniformly at random from $\cF$ to color the elements of $\binom{V'}{k-2}$

\vskip4pt
\State \textbf{pass 1:}
\Indent
	\State $(Q, S) \gets (\varnothing, \varnothing)$
	\For{each hyperedge $e = \{u_1,\ldots,u_k\}$ in the stream}
		\State let $i \in [k]$ be such that $u_i$ is the vertex in $e$ with the smallest ID
		\If{all $(k-2)$-sized subsets of $e\setm \{u_i\} \in E_{u_i}$ are colored the same}
			\State $Q \gets Q \cup \smash[b]{\left\{\big\{
				\flavor{u_i}{u_1},\ldots,\flavor{u_i}{u_{i-1}},
                                \flavor{u_i}{u_{i+1}},\ldots,\flavor{u_i}{u_k}\big\}\right\}}$
                                \Comment{store the implied shadow-hypergraph edge}
		\EndIf
		\If{$e \cap Z \ne \varnothing$}
			\State $S \gets S \cup \{e\}$
		\EndIf
	\EndFor
	\State create heavy/light oracle based on $S$
\EndIndent
\vskip2pt
\State \textbf{pass 2:}
\Indent
	\State $(A_L, A_H) \gets (0, 0)$
	\State let $Q^L$ be the set of hyperedges in $Q$ corresponding to light hyperedges in $E$
	\For{each hyperedge $e = \{u_1,\ldots,u_k\}$ in the stream}
		\If{$\oracle(e) = \light$}
			\State $A_L \gets A_L + \left|\left\{z\in V\,:\,\bigwedge_{i=1}^k\left\{
				\big\{\flavor{z}{u_1},\ldots,\flavor{z}{u_{i-1}},
                                \flavor{z}{u_{(i+1)}},\ldots,\flavor{z}{u_k}\big\}
				\subseteq Q^L\right\}\right\}\right|$
				\label{line:count-light}
		\Else
			\State $A_H \gets A_H + \sum_{i=1}^{k+1}\hat{\Nsimp}_e^i/i$, 
                                where $\hat{\Nsimp}_e^i := \left|\big\{z\in Z\,:\,e\cup\{z\}
                                \text{ a simplex with $i$ heavy edges} \big\}\right|$
				\label{line:count-heavy}
		\EndIf
	\EndFor
\EndIndent
\vskip2pt
\State \Return $A_L/p^\alpha + A_H/q$
\EndProcedure
\end{algorithmic}
\caption{Simplex counting in the ``meager'' case using $\tO\big(m/T^{2/(k+1)}\big)$ bits of space}
\label{alg:meager-best}
\end{algorithm*}

\paragraph{Analysis.}
We proceed in the now-established manner within our framework: we need to work
out the parameters $\alpha, \beta$, and $\gamma$. 

Consider a particular $k$-simplex $\Xi$ in $H$ and let the $(k-1)$-simplex
$\Xi'$ be its ``shadow simplex'' in $H'$, as in the above discussion. Then, if
$\Xi$ is light, it is detected at most once, when the hyperedge opposite to
its lowest-ID vertex arrives in the stream. For this detection to actually
happen, all edges of $\Xi'$ must be monochromatic, i.e., all
$\binom{k}{k-2}$ possible $(k-2)$-sized subsets of vertices in $\Xi'$ must
receive the same color. Let $p = 1/N$. The probability of this event is
$p^\alpha$, where
\[
  \alpha = \binom{k}{k-2} - 1 = \binom{k}{2} - 1 \,.
\]

Next, we observe as in \Cref{sec:ub-meager-better}, that for two distinct
simplices $\Xi_1$ and $\Xi_2$ in $H$, the events of their detection have
nonzero correlation iff they have an edge in common. Thus, the separation
property \eqref{eq:sep} holds. Moreover, when $\Xi_1 \cap \Xi_2 = \{e\}$, the
correlation is in fact nonzero only if both $\Xi_1$ and $\Xi_2$ have the same
minimum-ID vertex and that vertex lies in $e$. When this happens, the
simultaneous detection event occurs iff all $(k-2)$-sized subsets of shadow
vertices arising from vertices of $\Xi_1$ and $\Xi_2$ receive the same color.
Counting the number of such subsets shows that the probability of simultaneous
detection is $p^\beta$, for 
\[
  \beta = 2\binom{k}{k-2}-(k-1)-1 = 2\binom{k}{2}-k \,.
\]

Finally, an edge in $E'$ is made monochromatic (and thus, stored in $Q$) with
probability $p^{k-2}$ because $k-1$ different subsets must all be colored
identically. This gives $\gamma = k-2$, so 
\[
  \lambda 
  = \min\left\{\frac{2\binom{k}{2}-k}{\binom{k}{2}-1} - 1, 
    \frac{k-2}{\binom{k}{2}-1}\right\} 
  = \frac{k-2}{\binom{k}{2}-1} 
  = \frac{2}{k+1} \,,
\]
for $k \geq 3$. Thus, our algorithm runs in $\tO(m/T^{\lambda}) =
\tO\left(m/T^{2/(k+1)}\right)$ bits of memory in the worst case.

We have therefore proved the following result.

\begin{theorem} \label{thm:ub-meager-best}
  There is a $2$-pass algorithm that reads a stream of $m$ hyperedges
  specifying a $k$-uniform hypergraph $H$ and produces an
  $(\eps,o(1))$-estimate of $T_k(H)$; under the promise that $T_k(H) \ge T$,
  the algorithm runs in expected space $O(\eps^{-2} \log^2 n\cdot
  m/T^{2/(k+1)})$. \qed
\end{theorem}

\subsection{A One-Pass Algorithm} \label{sec:ub-meager-1pass}

Suppose we insisted on using at most one streaming pass. Can 
we hope to achieve sublinear space? The answer is yes, provided
the input $k$-graph is guaranteed not to have too much ``local 
concentration'' of simplices. This is captured in the parameters
$\Delta_E$ and $\Delta_V$ that were defined in \Cref{sec:prelim}.

This algorithm can be seen as a direct generalization of the very recent
algorithm $1$-pass triangle counting algorithm of Jayaram and
Kallaugher~\cite{JayaramK21} from graphs to hypergraphs.  We make clever use
of the fact that for each simplex, there is an implicit ordering of its
hyperedges given by the order in which they arrive in the input stream. When
the last hyperedge arrives, we hope to have sampled the remaining hyperedges
of the simplex.  If this event occurs, we will have found the simplex. 

As in \cite{JayaramK21}, we treat each vertex and each hyperedge as
\emph{active} or not, with the probabilities of ``being active'' being $p$ for
each vertex and $q$ for each hyperedge. To sample a hyperedge, we need both
itself to be active, as well as at least one of its constituent vertices
(endpoints) to be active. When a hyperedge $e$ arrives, we count how many
sampled hyperwedges are completed to a simplex by $e$. The details are given
in \Cref{alg:meager-1pass}. This algorithm does not feature heavy/light edge
partitioning.

%%% The algorithm float, placed here, works well in both one- and two-column formats!

\begin{algorithm*}[!th]
\begin{algorithmic}[1]
\Procedure{$k$-Simplex-Count-One-Pass}{$\sigma :$ stream of edges of $k$-graph $H = (V,E)$}

\State $p \gets 9\Delta_V/(\varepsilon^2 T)$;~
        $q \gets \max\{\Delta_E/\Delta_V, 1/\sqrt[k]{\Delta_V}\}$
        \Comment{probability that a vertex (resp.~an edge) is active}
\State let $f \colon V \to \{0,1\}$ be chosen randomly such that $\EE[f(v)] = p$ for all $v \in V$
\State let $g \colon E \to \{0,1\}$ be chosen randomly such that $\EE[g(e)] = q$ for all $e \in {E}$
\State $S \gets \varnothing$ \Comment{sampled edges}
\State $\tT \gets 0$ \Comment{accumulator for unbiased estimator}
\For{each hyperedge $e = \{u_1,\ldots,u_k\}$ in the stream}
	\For{each $v \in V$}
		\If{$f(v) = 1 \,\wedge\, (\{v\} \cup (e\setm \{u_i\}) \in S$ for each $i \in [k])$}
			\vskip2pt
			\State $\tT \gets \tT + 1/pq^k$ \label{line:one-pass-accum}
		\EndIf
	\EndFor
	\If{$g(e)\cdot \sum_{i=1}^k f(u_i) > 0$}
		\State $S \gets S \cup \{e\}$
	\EndIf
\EndFor
\State \Return $\tT$
\EndProcedure
\end{algorithmic}
\caption{Simplex counting using one streaming pass}
\label{alg:meager-1pass}
\end{algorithm*}

\paragraph{Analysis.}
Although we cannot quite use our analysis framework in plug-and-play fashion,
we can reuse most of the arguments from \Cref{sec:obliv-framework}.  Let
$\ell(\Xi_i)$ be the last edge to arrive in the stream for simplex $\Xi_i$.
Let $W_i$ be the hyperwedge in $\Xi_i$ that has $\ell(\Xi_i)$ as its base and
let $\apex(W_i)$ be the apex of $W_i$.  Then $\Xi_i$ is sampled if and only if
$\apex(\Xi_i)$ is active and all the edges of $W_i$ are sampled.  Since
$\apex(\Xi_i)$ is a common vertex to all the hyperedges in $W_i$, we only need
each hyperedge to be active. Therefore, the probability that $\Xi_i$ is
counted into $\tT$ is $pq^k$.

As before, let $\Lambda_i$ be the indicator random variable for whether or not
the $i$th simplex $\Xi_i$ is detected. Then
\[
  \tT = \frac{1}{pq^k} \sum_{i=1}^{T_k(H)} \Lambda_i
\]
by the logic of \cref{line:one-pass-accum}, so $\EE[\tT] = T_k(H)$, making $\tT$ an
unbiased estimator. For the variance, we have
\begin{align}
  \Var[\tT] 
  = \frac{1}{p^2 q^{2k}} \Var\left[\sum_{i=1}^{T_k(H)} \Lambda_i\right]
  = \frac{T_k(H)}{pq^k} + \sum_{i\neq j}\Cov[\Lambda_i,\Lambda_j] \,.
\end{align}
We can divide the pairs of simplices that have non-zero covariance into two categories:
\begin{itemize}
  \item Simplices $\Xi_i$ and $\Xi_j$ sharing an edge $e$. If $\apex(\Xi_i) =
  \apex(\Xi_j)$, then $W_i \cap W_j = \{e\}$, meaning that the probability
  that both simplices are found is $pq^{2k-1}$. If $\apex(\Xi_i) \neq
  \apex(\Xi_j)$, then $\ell(\Xi_i) = \ell(\Xi_j)$, so the probability that
  both simplices are found becomes $p^2 p^{2k-1} < pq^{2k-1}$. Using the
  same argument as in \cref{eq:s_cor_bound}, we can see that there are at
  most $T_k(H) \Delta_E$ pairs of simplices with covariance bound $pq^{2k-1}$.
  \item Simplices $\Xi_i$ and $\Xi_j$ sharing at least one vertex, but not
  an edge. In that case, if they share the apexes $\apex(\Xi_i) =
  \apex(\Xi_j)$, then $\EE[\Lambda_i \Lambda_j] = pq^{2k}$. If their apexes
  are not the same, then $\EE[\Lambda_i \Lambda_j] = p^2 q^{2k}$, so they are
  independent. Thus, there are at most $T_k(H) \Delta_V$ simplices with this
  covariance bound of $pq^{2k}$.
\end{itemize}

To bound the total error probability of our estimator, we use Chebyshev's
inequality to get
\begin{align}
  \Pr\left[|T_k(H) - \tT| \geq \eps T_k(H)\right] 
  &\le \frac{\Var[\tT]}{\eps^2 T_k(H)^2} \notag\\
  &\le \frac{1}{\eps^2 T_k(H)^2} \left(\frac{T_k(H)}{pq^k} +
    \frac{\Delta_E T_k(H)}{pq} +
    \frac{\Delta_V T_k(H)}{p}\right) \notag\\
  &\le \frac{1}{\eps^2 Tpq^k} + \frac{\Delta_E}{\eps^2 pqT} + 
    \frac{\Delta_V}{\eps^2 pT} \,,
\label{eq:one_pass_error_prob}
\end{align}
where at the final step we used $T_k(H) \ge T$.

We shall upper-bound each term in the expression
\eqref{eq:one_pass_error_prob} by $1/9$, thus bounding the error probability
by $1/3$. To control the third term, we set $p = 9\Delta_V/(\eps^2 T)$. To
control the first two terms, we now need $q = \max\left\{\Delta_E /
\Delta_V,\, 1/\sqrt[k]{\Delta_V}\right\}$. Since each edge is sampled with
probability $pq$, the overall space usage of \Cref{alg:meager-1pass}, given
the choices we made for $p$ and $q$, is $O(mpq \log n) = \tO((m/T) (\Delta_E +
\Delta_V^{1-1/k}))$. 

As usual, the error probability of the estimator can be brought down from
$1/3$ to an arbitrary $\delta$ by taking a median of $O(\log \delta^{-1})$
independent copies. Thus, we obtain the following theorem.

\begin{theorem} \label{thm:ub-meager-1pass}
  There is a $1$-pass streaming algorithm that $(\eps,\delta)$-estimates
  $T_k(H)$ using
  \[
    O\left(\frac{1}{\eps^2} \log\frac{1}{\delta} \log n \cdot
    \frac{m}{T}\left(\Delta_E+\Delta_V^{1-\frac{1}{k}}\right)\right)
  \]
  bits of space, under the promise that $T_k(H) \geq T$. \qed
\end{theorem}

\section{Lower Bounds} \label{sec:lb}

In this section, we prove some lower bounds on the space complexity of simplex
counting using a constant number of streaming passes; one of the results is
for $1$-pass algorithms. These bounds follow from well-known lower bounds in
communication complexity using the standard connection between streaming and
communication complexity. For our reductions from communication problems to
simplex counting, we use three different families of hypergraph constructions.
While these lower bounds are not the major contribution of this work, they do
play the important role of (a)~showing that some of our upper bounds have the
right dependence on the parameters $m, T, \Delta_E$, and $\eps$, at least in
certain parameter regimes, and (b)~clarifying where improvements might be
possible.

\subsection{Preliminaries} \label{sec:lb-prelim}

We set up some notation, provide some necessary background, and give a
high-level overview. The actual results and proofs appear in
\Cref{sec:lb-results}. We continue to use our standard notations $H$, $V$,
$E$, $k$, $m$, $n$, $T$, $\Delta_V$, $\Delta_E$, $T_k(H)$, $\eps$, and
$\delta$ from earlier sections.

\paragraph{A Subtlety in the Use of Asymptotic Notation.}
Most of our lower bounds will involve two or more parameters describing the
input, in sync with how the upper bounds involve multiple parameters.
Following standard practice in the literature in this area, we will use
$\Omega(\cdot)$-notation to state these lower bounds. It is worth taking a
moment to reflect on what this actually means, for there is a subtlety
here.%
\footnote{This subtlety appears in just about every past work on lower bounds
for triangle counting in graphs.}

The purpose of a lower bound is to clarify what {\em cannot} be achieved
algorithmically. Consider, for instance, an estimator for the number of
triangles in a graph that runs in $\tO(m/\sqrt{T})$ space. We would like to
claim that this combined dependence on $m$ and $T$ is near-optimal, in the
sense that there does not exist an algorithm that {\em guarantees} a space
bound of $o(m/\sqrt{T})$. One way to do this is to exhibit a family of graphs
$G$ such that (a)~approximating $T_3(G)$ for this family would require
$\Omega(N)$ space, thanks to some known $\Omega(N)$ communication complexity
lower bound, and (b)~$N = \Theta(m/\sqrt{T})$ for this family of graphs. We
would then state that estimating $T_3(G)$ requires $\Omega(m/\sqrt{T})$ space
(in fact, this is an actual result of \cite{CormodeJ17}). This is indeed how
almost all work in the area proceeds. We do the analogous thing for our
hypergraph problems.

Note, however, that such an $\Omega(m/\sqrt{T})$ lower bound {\em does not}
preclude the existence of algorithms with a different dependence between $m$
and $T$ that is in general incomparable with $m/\sqrt{T}$. Indeed, there {\em
are} algorithms for estimating $T_3(G)$ that run in $\tO(m^{3/2}/T)$
space~\cite{McGregorVV16,BeraC17} and for certain graph families, we {\em do}
have $m^{3/2}/T = o(m/\sqrt{T})$. The reason there is no contradiction here is
that those algorithms aren't guaranteeing a space bound of $o(m/\sqrt{T})$ for
{\em every} graph.

\paragraph{Communication Problems of Interest.}
Our streaming lower bounds make use of known lower bounds in communication
complexity and the well-known fact that a $p$-pass $S$-space streaming
algorithm can be simulated by a two-player $(2p-1)$-round communication
protocol with $S$-bit messages per round. In this work, we only use two-player
(Alice and Bob) communication problems for our reductions.
We shall always consider randomized communication complexity, allowing a
constant error probability, say $1/3$. 

In the {\em set disjointness} problem, $\disj_N$, Alice and Bob hold vectors
$\bx,\by \in \b^N$ (respectively), interpreted as characteristic vectors of
sets $S,T \subseteq [N]$. They must determine whether or not $S \cap T =
\varnothing$. We will use the variant of this problem called {\em unique
disjointness}, denoted $\udisj_N^{k}$, where it is promised that $|S| = |T| =
k$ and that $|S \cap T| \le 1$. The celebrated result of \cite{Razborov92}
shows that the communication complexity of $\udisj_N^{\ceil{N/4}}$ is
$\Omega(N)$.

In the {\em gap-disjointness} problem, $\gapdisj_{N,t,g}$, Alice and Bob
hold vectors $\bx,\by$ representing sets $S,T$ as above. They have to
distinguish between the cases $|S \cap T| \le \floor{t}$ and $|S \cap T| \ge
\ceil{t+g}$.  The $\Omega(N)$ lower bound on the \textsc{gap-hamming-distance}
problem~\cite{ChakrabartiR12}, combined with a folklore reduction, implies
that the communication complexity of $\gapdisj_{N,N/2,\sqrt{N}}$ is
$\Omega(N)$.

In the {\em indexing} problem, $\idx_N$, Alice holds a data vector $\bx \in
\b^N$ and Bob holds an index $y \in [N]$. They can only communicate one-way,
from Alice to Bob, after which Bob must output the $y$th data bit, $x_y$. It
is well known~\cite{Ablayev96} that the one-way communication complexity of
$\idx_N$ is $\Omega(N)$. In fact, this same lower bound holds under a {\em density
promise} where it is given that the Hamming weight of $\bx$ is $\ceil{N/2}$ (say).

\paragraph{Streaming Decision Problems.}
Our main problem, for which we seek lower bounds, is that of
$(\eps,\delta)$-estimating $T_k(H)$ given the hyperedge stream of a $k$-graph
$H$ with $n$ vertices and $m$ hyperedges that satisfies $T_k(H) \ge T$, for
some threshold $T$.  We shall make the additional mild (and very natural)
assumption that our algorithm for this task must output $0$ when $T_k(H) = 0$.
Every one of our algorithms satisfies this. We shall refer to this version of
the problem as $\SimplexEst_k(n,m,T,\eps)$.  We will target $\delta = 1/3$
throughout this section. 
It will be helpful to simplify the task required of our streaming algorithms
to decision tasks and to have compact notation to describe a task with many
parameters.

In the decision task $\SimplexDist_k(n,m,T)$, the goal is to decide (up to
error probability $1/3$) whether $T_k(H) = 0$ or $T_k(H) \ge T$, under the
promise that one of these is the case. If the task need only be solved under
promised upper bounds on $\Delta_E$ and $\Delta_V$, then we denote it
$\SimplexDist_k(n,m,T;\Delta_E,\Delta_V)$. Note that a $p$-pass $S$-space
algorithm for $\SimplexEst_k(n,m,T,1/2)$ immediately yields a $p$-pass
$S$-space algorithm for $\SimplexDist_k(n,m,T)$.

In the decision task $\SimplexSep_k(n,m,T,\eps)$, the goal is to decide (up to
error probability $1/3$) whether $T_k(H) \le (1-\eps)T$ or $T_k(H) \ge
(1+\eps)T$. Notice that a $p$-pass $S$-space streaming algorithm for
the problem $\SimplexEst_k(n,m,T,\eps/2)$ immediately yields a similar
algorithm for $\SimplexSep_k(n,m,T,\eps)$.

\subsection{Our Lower Bound Results} \label{sec:lb-results}

We now prove the many lower bounds informally stated in
\Cref{thm:lbs-informal}, one by one. Throughout, we assume that $k$ is a
constant and that $n$ and $m$ are large enough.

\paragraph{Sublinear Space Requires a Good Threshold.}
First, we justify requiring a threshold $T$ in all of our upper bound results:
without it, one cannot guarantee a sublinear-space solution.

\begin{theorem} \label{thm:lb-n}
  The task $\SimplexDist_k(n,m,1)$---and thus, the task
  $\SimplexEst_k(n,m,1,1/2)$---requires $\Omega(n^k/p)$ space in the worst
  case, given $p$ passes. As $m = O(n^k)$ for every $k$-graph, this precludes
  a sublinear-space $O(1)$-pass solution.
\end{theorem}
\begin{proof}
We reduce from $\udisj_N^{N/4}$, where $N = n^k$. Alice and Bob view their
respective inputs $\bx,\by \in \{0,1\}^{n^{k}}$ as labelings on $k$-dimensional hypercubes with
side length $n$. Based on their inputs, they construct certain set of 
hyperedges whose union produces a $k$-graph $H$.

The vertices of ${H}$ are split into 3 groups: $A = \{a_1,\ldots,a_n\}, B =
\{b_1,\ldots,b_n\}$ and $C=\{c_1,\ldots,c_n\}$. If $x[i_1,\ldots,i_{k}] = 1$, then
Alice inserts into ${H}$ the hyperedge $\{a_{i_1}, c_{i_2},\ldots,c_{i_{k}}\}$.
Bob subsequently follows the same procedure with his string $y$, except he
includes vertex $b_{i_1}$ instead of $a_{i_1}$ in his hyperedges. Bob also
includes in the stream the $n \binom{n}{k-2} = \Theta(n^{k-1})$ hyperedges of
the form $\{a_i, b_i\} \cup C'$ for all $(k-2)$-sized subsets $C' \subseteq C$.

We claim that $\bx$ and $\by$ intersect iff $H$ contains a $k$-simplex.
Indeed, if both vectors have a `$1$' entry at index $(i_1,\ldots,i_{k})$, then
the edges $\{a_{i_1}, c_{i_2},\ldots,c_{i_{k}}\}$, $\{b_{i_1},
c_{i_2},\ldots,c_{i_{k}}\}$ and all the edges containing $a_{i_1}, b_{i_1}$ and
the $(k-1)$ subsets of size $(k-2)$ of $C'' = \{c_{i_2},\ldots,c_{i_{k}}\}$ are
contained in ${H}$. This is $k+1$ hyperedges in total, forming a simplex.
Conversely, if $H$ contains a $k$-simplex, then the vertices of that simplex
must be some set of the form $\{a_{i_1}, b_{i_1}, c_{i_2}, \ldots  ,c_{i_{k}}\}$
because of the way ${H}$ was constructed. So that means that $x[i_1,\ldots,i_{k}]
= y[i_1,\ldots,i_{k}] = 1$.

The theorem follows by using the $\Omega(N) = \Omega(n^k)$ communication lower
bound for $\udisj_N^{N/4}$.
\end{proof}

\paragraph{A One-Pass Lower Bound.}
We now show that our $1$-pass algorithm from \Cref{sec:ub-meager-1pass} has
the correct dependence on the combination of $m$, $T$, and $\Delta_E$.

\begin{theorem} \label{thm:lb-delta-e}
  The task $\SimplexDist(n,m,T;\Delta_E,\Delta_V)$ requires $\Omega(m\Delta_E/T)$
  space using one pass.
\end{theorem}
\begin{proof}
We reduce from the problem $\idx_{n^k}$ under a density promise. The constructed hypergraph ${H}$ has a vertex set $V$ consisting of a ``matrix'' of vertices $\{v_{ij}\}_{i \in [k], j \in [n]}$ and $n^k$ other vertices $\{c_i\}_{i \in [n^k]}$.

Alice views her data vector $\bx \in \{0,1\}^{n^k}$ as a labeling of a $k$-dimensional hypercube. When $x[i_1,\ldots,i_k] = 1$, Alice inserts into $H$ the hyperedge $\{v_{1, i_1},\ldots,v_{k, i_k}\}$. Thus, she inserts as many edges as the Hamming weight of $\bx$, which is $\Theta(n^k)$.

Bob views his index $y \in [n^k]$ as an index into this hypercube: let $y = (t_1,\ldots,t_k)$. He then inserts the $kn^k$ hyperedges $\{c_i\} \cup X$ for all $i \in [n^k]$ and all $X \subset \{v_{1,t_1},\ldots,v_{k,t_k}\}$ with $|X| = k-1$. In other words, Bob inserts $n^k$ hyperwedges with apexes the $c$-vertices and base $\{v_{1,t_1},\ldots,v_{k,t_k}\}$. Overall, $m = \Theta(n^k)$ for the final ${H}$.

The final ${H}$ will have $T_k(H) = 0$ or $T_k(H) = T := n^k$, because all the hyperwedges Bob added will have a base that has either been made into a hyperedge by Alice or not. In the first case, the Bob is certain that $x[y] = 0$, because otherwise $n^k$ simplices would be formed with base $\{v_{1i_1},\ldots,v_{ki_k}\}$. In the latter case, Bob will be certain that $x[y] = 1$. Therefore, solving $\SimplexDist_k$ on $H$ does indeed solve $\idx_{n^k}$, completing the proof.
\end{proof}

\paragraph{Multi-Pass Lower Bounds in Terms of Hypergraph Parameters.}
Most of our algorithms use either two or four passes. Accordingly, we seek
lower bounds in terms of $m$, $T$, and $\Delta_V$ that work even against
multi-pass algorithms. We give three such lower bounds: one applicable to the
``abundant'' case (when $T$ is fairly large), one applicable to the ``meager''
case (when $T$ is not so large), and one applicable when too many simplices
concentrate at a single vertex.  All three bounds follow from the same
hypergraph construction, so we state them jointly.

\begin{theorem} \label{thm:lb-main}
Any algorithm solving $\SimplexDist_k(m,n,T;1,\Delta_V)$ on a hypergraph ${H}$ must use
\[
  \Omega\left(\min\left\{\frac{m^{1+1/k}}{T}, \frac{m}{T^{1-1/k}},
  \frac{m\Delta_V^{1/k}}{T}\right\}\right)
\]
bits of space in the worst case.
\end{theorem}
\begin{proof}
  We reduce from the $\udisj_n^{n/4}$ communication problem, with Alice's and
  Bob's inputs being $\bx,\by \in \b^n$. The constructed hypergraph ${H}$ has
  $(k+1)n$ vertices split into $(k+1)$ groups of $n$:
  \begin{align*}
    A_j &= \{a_1^j, \ldots, a_n^j\} \,, \quad\text{for each } j \in [k] \,; \\
    D &= \{d_1, \ldots, d_n\} \,.
  \end{align*}

  Corresponding to each index $i$ for which $x[i] = 1$, Alice inserts
  $n^{k-1}$ hyperedges into $H$: for each $(i_1,\ldots,i_{k-1}) \in
  [n]^{k-1}$, she inserts the hyperedge
  $\{d_i,a^1_{i_1},\ldots,a^{(k-1)}_{i_{k-1}}\}$.

  Bob does something more complicated. Corresponding to each index $i$ for
  which $y[i] = 1$, he adds $(k-1)n^{k-1}$ hyperedges, namely $\big\{d_i,
  a^1_{i_1},\ldots,a^{j-1}_{i_{j-1}},a^{j+1}_{i_{j+1}},\ldots,a^k_{j_k}\big\}$
  for $j=1,2,\ldots,(k-1)$. Additionally, Bob adds another $n^k$ hyperedges of
  the form $\{a^1_{i_1},\ldots,a^k_{i_k}\}$ to ${H}$: note that these are
  independent of $\by$. 
  A rough sketch of this construction, for $k=3$, is shown in
  \Cref{fig:lower_bounds_explanations}.

  \begin{figure}[tb] \centering
    \includegraphics[width=0.6\textwidth]{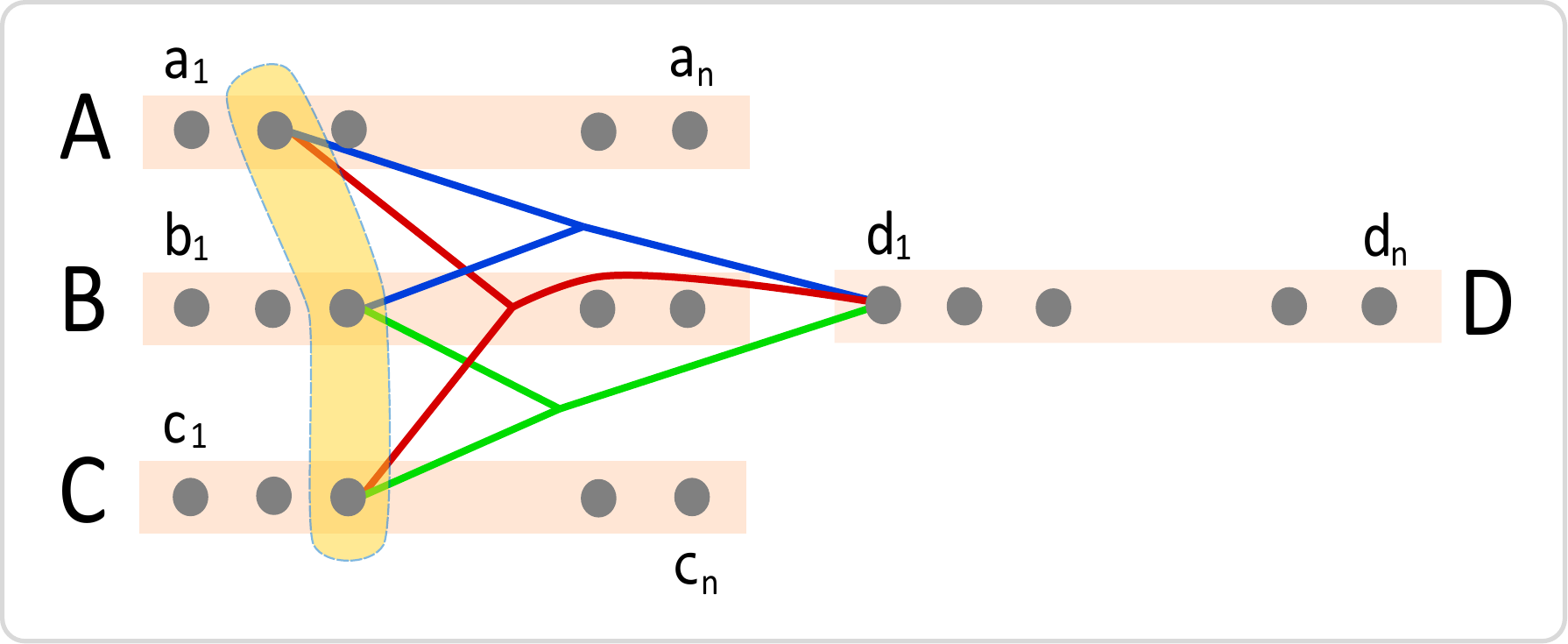}
    \caption{%
    The reduction in the proof of \Cref{thm:lb-main}, with $k=3$, $A_1 =A$,
    $A_2 = B$ and $A_3 = C$. The set $\{d_1,a_2,b_3,c_3\}$ is a simplex in $H$
    because $1 \in \bx \cap \by$, corresponding to vertex $d_1$. Alice inserts
    the blue edge $\{d_1,a_2,b_3\}$ while encoding $\bx$, while Bob inserts
    the red and green edges while encoding $\by$, plus the data-independent
    yellow edge. Similarly, $\{a_i,b_j,c_k,d_1\}$ are simplices for all $i,j,k
    \in [n]$.%
    }
    \label{fig:lower_bounds_explanations}
  \end{figure}

  In the end, ${H}$ has $m = \Theta(n^k)$ edges. Further, $T_k(H)$ is exactly $n^k$ times 
  the size of the intersection $|\bx \cap \by|$. Because of the unique intersection
  promise, we either have $T_k(H) = 0$ or $T_k(H) = T := n^k$, which proves the correctness
  of the reduction. Moreover, when $T_k(H) \ne 0$, we have $\Delta_V = n^k$.

  If there existed an algorithm $\mathcal{A}$ for the $\SimplexDist$ task guaranteed to use
  only $o(m^{1+1/k}/T)$ bits of memory, the resulting communication protocol for $\udisj_n^{n/4}$
  would use $o(n^{k+1} / n^k) = o(n)$ bits of communication, a contradiction.
  For the same reason, $\mathcal{A}$ could not have guaranteed a space bound of
  $o(m/T^{1-1/k})$ bits, nor $o(m\Delta_V^{1/k}/T)$ bits.
\end{proof}

\paragraph{A Lower Bound in Terms of the Approximation Parameter.}
Finally, we show that the inverse quadratic dependence on the approximation
parameter $\eps$, which appears in all our algorithms (and is indeed ubiquitous
in the literature on streaming algorithms) is necessary.

\begin{theorem} \label{thm:lb-eps}
  The task $\SimplexSep_k(n,m,T,\eps)$ requires $\Omega(\eps^{-2}/p)$ space
  using $p$ passes, even when $n$ and $m$ are functions of $\eps$ alone.
\end{theorem}
\begin{proof}
  We reduce from $\gapdisj_{n,n/2,\sqrt{n}}$.
  We repeat the hypergraph construction in the proof of \Cref{thm:lb-main}.
  As observed in that proof, the number of simplices $T_k(H)$ in that 
  hypergraph $H$ is exactly $n^k$ times the intersection size $|\bx \cap \by|$.
  Therefore, if $\eps$ and $n$ are related by $\eps = C/\sqrt{n}$,
  for a suitable constant $C$,
  then an algorithm for $\SimplexSep_k(n,m,T,\eps)$ solves $\gapdisj_{n,n/2,\sqrt{n}}$.
  The proof is completed by invoking known communication lower bounds.
\end{proof}

\section{Targeted-Sampling Algorithms Based on Hyperarboricity} \label{sec:subopt}

In this section, we discuss in more detail the more straightforward, but
trickier to analyze, targeted-sampling algorithm that was outlined in
\Cref{sec:tech-abundant}. It runs in space $\tO(m^{2-1/k}/T)$, which is
suboptimal, but still sublinear once $T$ is large enough. 

The algorithm, formalized as \Cref{alg:global-deg}, is similar to
\Cref{alg:abundant} in its general setup. The broad framework of its analysis
is also similar to what went before. However, two key insights required in the
analysis use mathematical tools from extremal hypergraph theory. Specifically,
we use the idea of hyperforest packing~\cite{FrankKK03} to define a new
quantity that we call the {\em hyperarboricity} of a hypergraph. We proceed to
give a generally applicable bound on this quantity and use it to bound the
space usage of our algorithm. We feel that the techniques and tools developed
here may be of independent interest and value.

\subsection{Definitions and Preliminaries}

For $F \subseteq E$, let $V(F)$ be the set of all vertices in $V$ that belong
to some hyperedge in $F$. For $S \subseteq V$, let $E(S)$ be the set of
hyperedges contained completely within $S$, i.e., the set of hyperedges {\em
induced} by $S$. 

Recall \cite{Bretto-book} that every hypergraph can be represented by a
bipartite graph. The {\em bipartite representation} of $H$ is the bipartite
graph $G_{H} = (V;E,R)$, where the edge $r = \{v,E\} \in R $ exists between a
vertex $v$ and a hyperedge $E$ of $H$ if $v \in E$. This $G_H$ construction
inspires the following definition.

In graph theory, a \emph{forest} is an acyclic graph. Equivalently, forests
are exactly the graphs in which for every subset $X$ of the edge set, the
number of incident vertices is strictly greater than $|X|$. In general
hypergraphs, keeping the bipartite representation in mind, we refer to this
property as the \emph{strong Hall property}. We shall use the corresponding
definition of ``hyperforests'' \cite{FrankKK03,Lovasz68,Lovasz70}.

\begin{definition}[Strong Hall property, hyperforest] \label{def:strong_hall}
  A hypergraph $H = (V,E)$ satisfies the strong Hall property if, for any
  non-empty subset $X \subseteq E$, we have that $|V(F)| > |F|$.
  Equivalently, for any $S \subseteq V$, we have that $|E(S)| < |S|$. 
  A hypergraph $H$ satisfying the strong Hall property is called a
  \emph{hyperforest}.
\end{definition}

Note that, unlike for usual graphs (i.e., $2$-graphs), acyclicity and the
strong Hall property do not coincide.

The arboricity of an ordinary graph is the minimum number of forests that its
edge set can be partitioned into. The above notion of hyperforest provides a
natural generalization of this quantity to hypergraphs.

\begin{definition}[Hyperarboricity]
  The {\em hyperarboricity} of a hypergraph $H$, denoted $\rho(H)$, is the 
  minimum number of edge-disjoint hyperforests that $E(H)$ can be decomposed 
  into. 
\end{definition}

Seeking a generalization of the classic forest-packing theorems of Tutte and
Nash-Williams \cite{ChenMWZZ94}, Frank, Kiraly, and Kriesell~\cite{FrankKK03}
use matroids to give the following necessary and sufficient condition for the
decomposition of a hypergraph into edge-disjoint sub-hyperforests.

\begin{lemma}[Hyperforest packing; Theorem 2.10 of~\cite{FrankKK03}]
\label{lem:hyperforest-packing}
  The edge-set $E$ of a hypergraph $H = (V,E)$ can be decomposed into $k$
  hyperforests iff $|E(X)| \leq k(|X|-1)$ for every non-empty
  subset $X \subseteq V$. \qed
\end{lemma}

Using this, we obtain an extremal condition that characterizes the
hyperarboricity concept.

\begin{lemma} \label{lem:hyperarb-ratio}
  For every hypergraph $H$,
  \begin{align*}
    \rho(H) = \max_{H' = (V',E') < H}\left\lceil\frac{|E'|}{|V'|-1}\right\rceil \,,
  \end{align*}
  where ``$<$'' denotes the sub-hypergraph relation.
\end{lemma}
\begin{proof}
  Let $V' \subseteq V$ and consider an arbitrary sub-hypergraph $H' = (V',E')$
  of the induced (by $V'$) sub-hypergraph $H'' = (V',E'') < H$. Then, $|E'|
  \leq |E''|$, so by \Cref{lem:hyperforest-packing},
  \[
    \rho(H) \geq \frac{|E''|}{|V''|-1} \geq \frac{|E'|}{|V'|-1} \,.
  \]
  Since $\rho(H)$ is an integer, we arrive at our conclusion.
\end{proof}

\subsection{A General Upper Bound on Hyperarboricity}
\label{sec: hyperarboricity_bounds}

In analyzing our eventual algorithm, a key portion of the combinatorial
argument will need to bound a sum of the form $s_1$, shown in
\cref{eq:sum-min-deg}. A classic result of Chiba and Nishizeki~\cite{ChibaN85}
bounds this sum in the special case when $H$ is an ordinary graph. For our
$k$-uniform hypergraph generalization, we start with the following useful
upper bound on the hyperarboricity of such a hypergraph.

\begin{lemma} \label{lem:hyperarb-bound}
  The hyperarboricity $\rho(H)$ of an $m$-edge $k$-graph $H$ is bounded by
  $\rho(H) = O\left(m^{(k-1)/k}\right)$.
\end{lemma}
\begin{proof}
  Invoking \Cref{lem:hyperarb-ratio}, we let $H' < H$ be a sub-hypergraph
  satisfying
  \[
    \rho(H) = \left\lceil\frac{|E(H')|}{|V(H')|-1}\right\rceil \,.
  \]
  Let $p = |V(H')|$, $q = |E(H')|$ and let $s = \binom{p}{k}$ be the number of
  edges on a complete $k$-uniform hypergraph with $p$ vertices. We make two
  algebraic observations.
  \begin{itemize}
    \item First, we observe that $p = O((s+p)^{1/k})$. 
    This is because $s \geq Cp^k \geq Cp^k -p$ for some constant $C$, which
    implies that $s+p = \Omega(p^k)$.
    \item Second, we observe that $k!\cdot s = O((p-1)^k)$. 
    This holds because $k!\cdot s = p(p-1)\cdots(p-k+1) \leq p(p-1)^{k-1} =
    (p-1)^k + (p-1)^{k-1} = O((p-1)^k)$.
  \end{itemize}
  We now branch into two cases.

  In the first case, $s\leq m$. Then, we have
  \begin{align*}
    \rho(H) = \left\lceil\frac{q}{p-1}\right\rceil
    \leq \left\lceil\frac{s}{p-1}\right\rceil
    &\leq \left\lceil\frac{p(p-2)\cdots(p-k+1)}{k!} \right\rceil \\
    &\leq \left\lceil\frac{p^{k-1}}{k!}\right\rceil
    = O(p^{k-1}) \,.
  \end{align*}
  By the first observation, $p = O((s+p)^{1/k})$, giving
  \[
    \rho(H) = O\left((s+p)^{(k-1)/k}\right)
    = O\left((m+n)^{(k-1)/k}\right)
    = O\left(m^{(k-1)/k}\right)
  \]
  because $s \leq m, p\leq n$ and $n = O(m)$.

  In the second case, $s \geq m$. By the second observation above, we get:
  \begin{align*}
     \rho(H) 
     = \left\lceil\frac{q}{p-1}\right\rceil 
     &\le \left\lceil\frac{m}{p-1}\right\rceil\\
     &\stackrel{m \leq s}{\leq} \left\lceil\frac{ms^{\frac{1}{k-1}}}{(p-1)^{\frac{k}{k-1}}}\right\rceil^{\frac{k-1}{k}} \\
     &\leq(k!)^{-\frac{1}{k}}\left( \frac{m(p-1)^{\frac{k}{k-1}}}{(p-1)^{\frac{k}{k-1}}}\right)^{\frac{k-1}{k}}
     = O\left(m^{\frac{k-1}{k}}\right) \,. \qedhere
  \end{align*}
\end{proof}

As a corollary to the above bound on hyperarboricity, we obtain the following 
combinatorial fact, which will be useful in our algorithmic analysis later on.

\begin{lemma} \label{thm:sum_min_deg_bound}
  Let $H = (V,E)$ be a $k$-uniform hypergraph with $|E| = m$ and $|V| = n = O(m)$. Then:
  \begin{align}
    \sum_{\{u_1,\ldots,u_k\}\in E} \min\{\deg(u_1),\ldots,\deg(u_k)\} \leq km\rho(H) = O(m^{2-1/k}) \,.
  \end{align}
\end{lemma}
\begin{proof}
  Consider a decomposition of $H$ into $\rho(H)$ hyperforests $\cF_i = (V_i,
  E_i)$ for $i \in [\rho(H)]$. In each hyperforest $\cF_i$, we will associate
  each hyperedge $e \in E_i$ with a \emph{representative vertex} $v^{(i)}(e)
  \in E$ such that no two hyperedges in $E_i$ share a representative. We can
  find such a mapping from $E_i$ to $V_i$ if and only if $G_{\cF_i}$ has a
  perfect matching that saturates $E_i$. Since $\cF_i$ is a hyperforest,
  \Cref{def:strong_hall} and Hall's Matching Theorem ensure that a perfect
  matching saturating $E_i$ exists. Using \Cref{lem:hyperarb-bound} and the
  handshake lemma, we get:
  \begin{align*}
    \sum_{\{u_1,\ldots,u_k\}\in E}\min\{\deg(u_1),\ldots,\deg(u_k)\}
    &= \sum_{1\leq i \leq \rho(H)}\sum_{\{u_1,\ldots,u_k\} \in E_i}\min\{\deg(u_1),\ldots,\deg(u_k)\} \\
    &\le \sum_{1\leq i \leq \rho(H)}\sum_{e \in E_i} \deg(v^{(i)}(e)) \\
    &\le \sum_{1\leq i \leq \rho(H)}\sum_{v \in V}\deg(v) = km\rho(H) = O\left(m^{2-1/k}\right) \,. \qedhere
  \end{align*}
\end{proof}
\begin{remark}
  This bound is tight. Equality holds for complete $k$-graphs.
\end{remark}

%%%%%%%%%%%%%%%%%%%%%%%%%%%%%%%%%%%%%%%%%%%%%%%%%%%

\subsection{The Algorithm Based on Global Degree Ordering}

We are ready to give the details of our algorithm.  Having seen an outline of
the technique in \Cref{sec:tech-abundant} and a detailed description of a more
elaborate algorithm in a similar vein in \Cref{sec:ub-abundant}, we are prepared
to dive right in to the formal description. This is given in
\Cref{alg:global-deg}.  The main new piece of notation here is ``$\prec$''
which denotes a {\em total ordering} on $V$ according to the degrees of the
vertices, settling ties by comparing IDs. 

\begin{algorithm*}[!ht]
\begin{algorithmic}[1]
\Procedure{$k$-Simplex-Count-Easy}{$\sigma$}

\State \textbf{pass 1:} \Comment{$O(1)$ space}
\Indent
	\State pick an edge $e = \{u_1,u_2,\ldots,u_k\} \in E$ using reservoir sampling
\EndIndent
\State \textbf{pass 2:} \Comment{$O(1)$ space}
\Indent
	\State calculate the degrees $\deg(u_1), \deg(u_2),\ldots,\deg(u_k)$
\EndIndent
\State \textbf{pass 3:} \Comment{$\tO(R)$ space}
\Indent
	\State $R \gets \left\lceil \min\{\deg(u_1), \deg(u_2),\ldots,\deg(u_k)\}\cdot m^{-1/2} \right\rceil$
	\State re-arrange (if needed) $u_1,u_2,\ldots,u_k$ so that $u_1 \prec u_2 \prec \cdots \prec u_k$
	\For {$j=1$ to $R$, in parallel,}
		\State $Z_j \gets 0$
		\State construct a virtual stream $\sigma_j$ of the vertices in 
                        $\Nhd(u_1) \setm \{u_2,\ldots,u_k\}$ \label{line:virtual-stream}
		\State sample vertex $x_j$ from $\Nhd(u_1) \setm \{u_2,\ldots,u_k\}$ u.a.r.,
			using $\ell_0$-sampling on $\sigma_j$ \label{line:global-deg-l0-sample}
	\EndFor
	\State in parallel, approximate $|\Nhd(u_1)|$ using $F_0$-estimation
\EndIndent
\State \textbf{pass 4:} \Comment{$O(R)$ space}
\Indent
	\State compute the degrees $\deg(x_1),\deg(x_2),\ldots,\deg(x_R)$
	\For {$j=1$ to $R$, in parallel,}
    	\If {$u_k \prec x_j$ \textbf{and} $\{u_1,u_2,\ldots,u_k,x_j\}$ form a $k$-simplex}
                \label{line:valid-label}
        	\State $Z_j \gets |\Nhd(u_1)|-k+1$ \Comment{$u_1$ has the smallest degree}
	    \EndIf
	\EndFor
\EndIndent
\vskip4pt
\State \Return $Y \gets (m/R) \sum_{j=1}^R Z_j$ \Comment{Average out the trials and scale}
\EndProcedure
\end{algorithmic}
\caption{Counting simplices in $k$-uniform hypergraph streams the easy (and suboptimal) way} \label{alg:global-deg}
\end{algorithm*}
%%%%%%%%%%%%%%%%%%%%%%%%%%%%%%

One subtlety in this algorithm, which does not show up in the corresponding
(and much simpler) algorithm for graphs~\cite{BeraC17}, is the uniform
sampling of the vertex in the neighborhood of $u_1$
(lines~\ref{line:virtual-stream}--\ref{line:global-deg-l0-sample}). In graphs,
reservoir sampling suffices because sampling a neighbor of $u_1$ is equivalent
to sampling uniformly an edge that is incident on $u$. The same is not true
for hypergraphs, because sampling a hyperedge uniformly from $\Nhd(u_1)$ gives
a probabilistic advantage to vertices that share many hyperedges with $u_1$.
Thus, we need to use $\ell_0$-sampling techniques in a substream consisting of
neighbors of $u_1$.

A second caveat that appears is in the calculation of the size of the
neighborhood of $u_1$. Treating every edge containing $u_1$ as an element in a
virtual stream $\sigma_{u_1}$, we can calculate $|\Nhd(u_1)|$ by counting the
distinct elements in $\sigma_{u_1}$. This can be done either precisely, in
$O(n)$ space, or approximately in $\tO(1)$ space, using an $F_0$-estimation
algorithm.

\subsection{Analysis of the Algorithm} \label{sec:subopt-analysis}

Our analysis proceeds almost exactly like \Cref{sec:abundant-anal} The only
differences are the presence of $\ell_0$-sampling, $F_0$ estimation, and the
use of the upper bounds in \Cref{sec: hyperarboricity_bounds}.

\begin{remark}
  In the analysis below, we will assume that the algorithms for $\ell_0$
  sampling and $F_0$ estimation produce the desired outputs correctly. Any
  precision errors are assumed to be absorbed into the quality of the
  estimator $Y$. Our analysis does not examine the implications of this
  assumption in a more pedantic way so that we can focus on the salient points
  of the algorithm.
\end{remark}

\paragraph{Unbiased Estimator.}
First we prove that our estimator $Y$ is unbiased. Define the {\em label} of a
simplex $\Xi$ to be $(e,x)$ where $x$ is the $\prec$-maximal vertex in $\Xi$
and $e = \Xi \setm \{x\}$. For each $e \in E$, let $\nsimp_e$ denote the
number of simplices with a label of the form $(e,\cdot)$.  Let $\cE_e$ denote
the event that the hyperedge sampled in pass~1 is $e$. When $\cE_e$ occurs,
the $\ell_0$-sampling in \Cref{line:global-deg-l0-sample} of
\Cref{alg:global-deg} causes us to pick one of the $|\Nhd(u_1)|-k+1$ vertices
in $\Nhd(u_1) \setm e$ uniformly at random and exactly $\nsimp_e$ of these
vertices complete a valid simplex label with hyperedge $e$ (per the logic in
\cref{line:valid-label}).  Therefore, 
\begin{align}
  \EE[Z_k\mid{\cE}_{e}] = \nsimp_e\frac{1}{|\Nhd(u_1)|-k+1}(|\Nhd(u_1)|-k+1) = \nsimp_e \,,
\end{align}
and so, by the law of total expectation,
\begin{align}
  \EE[Y] = \frac{m}{R}\sum_{j=1}^R\sum_{e \in E}\frac{1}{m}\EE[Z_j \mid{\cE}_e]
  = \frac{1}{R}\sum_{j=1}^R \sum_{e \in E} \nsimp_e
  = T_k(H) \,.
\end{align}

%%%%%%%%%%%%%%%%%%%%%%%%%%%%%

\paragraph{Bounded variance.}
Now we attempt to bound the variance. Proceeding along similar lines as above,
\begin{align*}
  \EE[Z_j^2\mid \cE_e] &= (|\Nhd(u_1)|-k+1)\cdot \nsimp_e \\
  &\le |\Nhd(u_1)|\cdot \nsimp_e \,, \quad\text{and} \\
  \EE[Z_{j_1}Z_{j_2}\mid\cE_{e}] &= \nsimp_e^2 \,,
\end{align*}
because $Z_{j_1}$ and $Z_{j_2}$ are independent events even when conditioned
on $\cE_{e}$. So we have
\begin{align*}
  \EE[Y^2\mid\cE_e]
  &= \EE\left[\left(\frac{m}{R}\sum_{k=1}^R Z_k\right)^2\, \middle|\,
    \cE_{e}\right] \\
  &= \frac{m^2}{R^2}\sum_{j=1}^R \EE[Z_j^2 \mid \cE_{e}] +
    \frac{m^2}{R^2}\sum_{j_1\neq j_2} \EE[Z_{j_1}Z_{j_2} \mid \cE_{e}] \\
  &\leq \frac{m^2 |\Nhd(u_1)| \cdot \nsimp_e}{R} + \frac{m^2 R(R-1) \nsimp_e^2}{R^2} \\
  &\leq (k-1) m^{5/2} \nsimp_e + m^2 \nsimp_e^2
\end{align*}
where the final step uses $|\Nhd(u_1)| / R \le (k-1) \deg(u_1) / R \leq
(k-1)\sqrt{m}$, by the definition of $R$. Next, using the law of total
expectation, we obtain
\begin{align}
  \Var[Y] \leq \EE[Y^2]
  &= \sum_{e \in E} \frac{1}{m} \EE[Y^2\mid\cE_{e}] \notag\\
  &\leq (k-1) m^{3/2} \sum_{e \in E} \nsimp_e + m\sum_{e \in E} \nsimp_e^2 \notag\\
  & \leq (k-1) m^{3/2} T_k(H) + m\sum_{e \in E} \nsimp_e^2 \,.  \label{eq:var-bound-subopt}
\end{align}

At this point we need a bound on $\sum \nsimp_e^2$. To that end, we will use the following lemma.
\begin{lemma} \label{lem:eden_generalization}
  Let $S_{u_1}:=\{w \in \Nhd(u_1) \mid u_1\prec w\}$ be the set of neighbors of
  $u_1$ that are \emph{after} it in the total ordering of vertices based on
  degree. Then $|S_{u_1}| = O\left(\sqrt{m}\right)$.
\end{lemma}
\begin{proof}
  We know that for each $w\in S_{u_1}$,
  \[
    \deg(w) \geq \deg(u_1) \geq \frac{|\Nhd(u_1)|}{k-1} 
    \geq \frac{|S_{u_1}|}{k-1} \,,
  \]
  Thus, by the handshake lemma, we have
  \begin{align*}
    km \geq \sum_{w\in S_{u_1} } \deg(w) 
    \geq |S_{u_1}|\left(\frac{|S_{u_1}|}{k-1}\right) \,,
  \end{align*}
  implying $|S_{u_1}| = O(\sqrt{m})$.
\end{proof}

Using \Cref{lem:eden_generalization}, we obtain that 
\begin{align}
  \sum_{e\in E} \nsimp_e^2 = O\left(T_k(H)\sqrt{m}\right)
  \label{eq:sum-square-bound}
\end{align} 
because, for each edge $e$, we have that $\nsimp_e \leq |S_{u_1}| =
O(\sqrt{m})$, where we use $S_{u_1}$ as in the proof of
\Cref{lem:eden_generalization}. Finally, plugging the
bound~\eqref{eq:sum-square-bound} into \cref{eq:var-bound-subopt}, we conclude
that
\begin{align}
  \Var[Y] = O\left(m^{3/2}T_k(H)\right) \,.
\end{align}

\paragraph{Space Complexity.}
Using the \emph{median-of-means} improvement, given as \Cref{lem:median-of-means}, we
can see that our algorithm uses space $O(m^{3/2} B / T)$ to output an $(\eps,
1/3)$-approximation to $T_k(H)$, where
\begin{align}
  B = O(\text{space}(F_0\text{-estimation}) + R \cdot \text{space}(\ell_0\text{-sampling}))
  \label{eq:b-formula}
\end{align}
is the space used by one instance of \Cref{alg:global-deg}. For simplicity,
assume that the $F_0$-estimation and $\ell_0$-sampling procedures function
perfectly (this does not affect the ultimate asymptotic results). Using known
results and hiding the dependence on accuracy parameters, we have
space$(F_0$-estimation$) = O(\log n)$ and space$(\ell_0$-sampling$) = O(\log^2
n)$. It remains to bound $\EE[R]$. We have that
\begin{align*}
  \EE[R] 
  &= \frac{1}{m} \sum_{\{u_1,\ldots,u_k\}\in E}
    \left\lceil \frac{\min\Big\{\deg(u_1),\ldots,\deg(u_k)\Big\}}{\sqrt{m}} \right\rceil \\
  &= O\left(m^{-3/2}\sum_{\{u_1,\ldots,u_k\}\in E}
    \min\Big\{\deg(u_1),\ldots,\deg(u_k)\Big\}\right) \,.
\end{align*}
This is where our work using hyperarboricity comes in. Using \Cref{thm:sum_min_deg_bound}, we get that
\[
  \EE[R] 
  = O\left(m^{2-\frac{1}{k}-\frac{3}{2}}\right) 
  = O\left(m^{\frac{1}{2}-\frac{1}{k}}\right) \,.
\]
Plugging this into \Cref{eq:b-formula}, we get
\begin{align*}
  \EE[B] 
  = O\left(\log n + m^{\frac{1}{2}-\frac{1}{k}}\log^2 n\right) 
  = \tO\left(m^{\frac{1}{2}-\frac{1}{k}}\right) \,.
\end{align*}

We have thus arrived at the following theorem.

\begin{theorem} \label{thm:ub-abundant-subopt}
  There is a $4$-pass streaming algorithm that $(\eps, \delta)$-approximates
  the number of simplices in a $k$-uniform hypergraph $H$ using
  $\tO\big(m^{2-1/k}/T\big)$ space, under the promise that
  $T_k(H) \geq T$. \qed
\end{theorem}

\section{Concluding Remarks}

We initiated a systematic study of the problem of counting simplices in a
hypergraph given as an edge stream, a natural generalization of the
much-studied triangle counting problem. We obtained several sublinear-space
algorithms for the problem: which of them is best depends on some
combinatorial parameters of the problem instance.  Overall, we learned that
established methods for triangle and substructure counting in graph streams
are not by themselves enough and considerable effort is required to deal with
the more intricate structures occurring in hypergraphs.

In some parameter regimes (what we called the ``abundant'' case), we obtained
provably space-optimal algorithms. However, in the ``meager'' case, seeking
algorithms with space complexity of the form $m/T^{\lambda}$, we established
an upper bound with $\lambda = 2/(k+1)$ and a lower bound with $\lambda =
1-1/k$. Closing this gap seems to us to be the most compelling follow-up
research question.

The simplex counting problem has the potential to impact research on pattern
detection and enumeration. As triangle counting has done in the past two
decades, simplex counting may offer new insights on how sampling techniques
can exploit the structure of graphs and hypergraphs to extract meaningful
information.

\bibliographystyle{alpha}
\newcommand{\etalchar}[1]{$^{#1}$}

\end{document}